\documentclass[11pt]{article}
\usepackage{paralist}
\usepackage{color}
\usepackage{bm}

\usepackage{soul}

\usepackage[cm]{fullpage}

\usepackage{amsfonts}
\usepackage{amssymb}

\usepackage{amsmath}
\usepackage{constants}
\usepackage{amsthm}
\makeatletter
\newtheorem*{rep@theorem}{\rep@title}
\newcommand{\newreptheorem}[2]{%
\newenvironment{rep#1}[1]{%
 \def\rep@title{#2 \ref{##1}}%
 \begin{rep@theorem}}%
 {\end{rep@theorem}}}
\makeatother

\newcount\shortyear\newcount\shorthour\newcount\shortminute
\shorthour=\time\divide\shorthour by 60\shortyear=\shorthour
\multiply\shortyear by 60\shortminute=\time\advance\shortminute by
-\shortyear \shortyear=\year\advance\shortyear by -1900

\def\zeit{\number\shorthour:\ifnum\shortminute<10 0\number\shortminute
\else\number\shortminute\fi}

\usepackage{ifpdf}
\newcommand{\mydriver}{hypertex}
\ifpdf
 \renewcommand{\mydriver}{pdftex}
\fi
\usepackage[breaklinks,\mydriver]{hyperref}

\usepackage[margin=1in]{geometry}

\newcommand{\vol}{\textrm{vol}}

\newcommand{\pot}{\textrm{pot}}
\newcommand{\poly}{\textrm{poly}}
\newcommand{\LL}{{\bf \mathcal{L}}}

\newcommand{\rem}{\textrm{rem}}
\newcommand{\reg}{\textrm{reg}}
\newcommand{\p}{\textbf{p}}

\newcommand{\q}{\textbf{q}}

\newcommand{\vv}{\textbf{v}}

\newcommand{\1}{\textbf{1}}

\newcommand{\D}{\textbf{D}}
\newcommand{\A}{\textbf{A}}
\newcommand{\W}{\textbf{W}}
\newcommand{\I}{\textbf{I}}

\newcommand{\red}{\textcolor{red}}

\theoremstyle{plain}
\newtheorem{theorem}{Theorem}[section]%[chapter]
\newtheorem{fact}[theorem]{Fact}
\newtheorem{lemma}[theorem]{Lemma}

\newtheorem{claim}[theorem]{Claim}

\newreptheorem{theorem}{Theorem}
\newreptheorem{lemma}{Lemma}
\newtheorem{definition}[theorem]{Definition}

\theoremstyle{definition}

\newenvironment{remark}[1][Remark.]{\begin{trivlist}
\item[\hskip \labelsep {\bfseries #1}]}{\end{trivlist}}

\newcommand{\centr}{\ensuremath{\Delta}}

\newcommand{\junk}[1]{{}}

 \providecommand{\norm}[1]{\lVert#1\rVert}
\newconstantfamily{c}{symbol=c}

\newconstantfamily{smallconst}{symbol=\kappa}

\title{Testing Cluster Structure of Graphs}

\author{Artur Czumaj\footnote{Department of Computer Science and Centre for Discrete Mathematics and its Applications (DIMAP), University of Warwick. Supported in part by DIMAP and by EPSRC grant EP/J021814/1. Email: \url{A.Czumaj@warwick.ac.uk}.}
\and
Pan Peng\footnote{Department of Computer Science, TU Dortmund; State Key Laboratory of Computer Science, Institute of Software, Chinese Academy of Sciences. Supported by ERC grant No. 307696.
Email: \url{pan.peng@tu-dortmund.de}.}
\and
Christian Sohler\footnote{Department of Computer Science, TU Dortmund. Supported by ERC grant No. 307696.
Email: \url{christian.sohler@tu-dortmund.de}.}
}

\date{}
\begin{document}

%===========================================================================

\begin{titlepage}

\maketitle

\thispagestyle{empty}

%\Artur{I changed (commented) \emph{theoremstyle definition} (it looked ugly and poorly readable in my opinion).}

\begin{abstract}
We study the problem of recognizing the cluster structure of a graph in the framework of property testing in the bounded degree model. Given a parameter $\varepsilon$, a $d$-bounded degree graph is defined to be \emph{$(k, \phi)$-clusterable}, if it can be partitioned into no more than $k$ parts, such that the (inner) conductance of the induced subgraph on each part is at least $\phi$ and the (outer) conductance of each part is at most $c_{d,k}\varepsilon^4\phi^2$, where $c_{d,k}$ depends only on $d,k$.
%$O(\varepsilon^4\phi^2)$, assuming that $k$ and $d$ are constant.
%\Artur{Notation $O_{d,k}()$ is not very frequently used and should be avoided in the abstract.}
%\Artur{I do have a problem with using notation $\widetilde{O}_{d,k}$ in the abstract since it's not standard. However, I'm not sure if skipping the dependency on $k$ and $\varepsilon$ would make it much better. {\bf Pan:} I tried to use $c_{d,k}$ now so that we can avoid using $\widetilde{O}_{d,k}$ or ${O}_{d,k}$. \textbf{Artur:} OK --- let's use these constants $c_{k,d}$.}
%\Pan{I changed this, but I am not sure if it is better now}
Our main result is a sublinear algorithm with
%\Artur{Discussion in the abstract about both query complexity and running time doesn't make sense. In the paper we should talk about one of these notions, and so (rather arbitrarily) I decided to to stick to the running time.}
%query complexity and
the running time
$\widetilde{O}(\sqrt{n}\cdot\poly(\phi,k,1/\varepsilon))$
%$\widetilde{O}(\sqrt{n}\cdot\poly(1/\varepsilon))$
%
that takes as input a graph with maximum degree bounded by $d$, parameters $k$, $\phi$, $\varepsilon$, and with probability at least $\frac23$, accepts the graph if it is $(k,\phi)$-clusterable and rejects the graph if it is $\varepsilon$-far from $(k, \phi^*)$-clusterable for
$\phi^* = c'_{d,k}\frac{\phi^2 \varepsilon^4}{\log n}$, where $c'_{d,k}$ depends only on $d,k$.
%$\phi^* = O(\frac{\phi^2 \varepsilon^4}{\log n})$.
%
By the lower bound of $\Omega(\sqrt{n})$ on the number of queries needed for testing graph expansion, which corresponds to $k=1$ in our problem, our algorithm is asymptotically optimal up to polylogarithmic factors.
%
%Our algorithm is the first property tester that tests pairwise closeness of distributions of random walks starting from a pair of sample vertices and draws from that conclusions on the graph structure.
%\Artur{I don't think we need this sentence in the abstract.}
\end{abstract}

\end{titlepage}

%===========================================================================

\section{Introduction}

\emph{Cluster analysis} is a fundamental task in data analysis that aims to partition a set of objects into maximal subsets (called \emph{clusters}) of similar objects. In \emph{graph clustering}, the objects to be clustered are the vertices of a graph and the edges of a graph describe relations between them. These relations may have interpretations for data analysis. For example, if the graph is the friendship graph of a social network, i.e., the vertices are the users of a social network and the edges correspond to friendship relations, edges may indicate that the users are socially related and/or have similar interests. In a co-author graph, where the vertices are authors and edges describe co-authorships, edges may be interpreted as a sign that the authors work in the same scientific community. A \emph{cluster} is then a maximal subset of vertices that are \emph{well-connected} to each other, where the precise meaning of being well-connected can be defined in various ways.

In many cases, once we know the interpretation of a single edge, there is a natural interpretation of clusters. For example, clusters in a friendship graph correspond to social groups or clusters in a co-author graph correspond to scientific communities. For similar reasons, a vast amount of graph clustering methods are applied to many different kinds of social/information/biological networks to reveal hidden cluster structure, etc. (see, e.g., surveys \cite{For10:community,POM09:communities,Sch07:clustering}).

Many efficient algorithms for finding clusters in a graph have been developed. However, with the increasing focus on the study of very large networks, we have to concentrate on new features of the clustering algorithms. For example, if one tries to find clusters in the World Wide Web or in a big social network, even linear time algorithms might be too slow. This is particularly important if one wants to study the temporal development of the clusters, which require to solve the problem on many instances (each for a different point of time). In such cases, we need \emph{sublinear time algorithms}. We develop such an algorithm in this paper. Our algorithm can be used to test, if a given graph has a cluster structure, i.e., is composed of at most $k$ clusters.
%Furthermore, if the graph is composed of $k$ clusters, the algorithm outputs at least one vertex from every big cluster (note that we cannot output a complete clustering of the graph as this requires linear time).

We will develop the algorithm in the framework of \emph{Property Testing} for bounded degree graphs \cite{GR02:testing}. In this framework, an algorithm has oracle access to an undirected graph $G=(V,E)$ with a bound $d$ on the maximum degree, with $d$ typically assumed to be constant. An algorithm is called a \emph{property tester for a given property $\Pi$} (in our case, the property of all graphs that have a cluster structure with at most $k$ clusters), if it accepts with probability at least $\frac23$ every graph that has the property $\Pi$ and rejects with probability at least $\frac23$ every graph that is \emph{$\varepsilon$-far from $\Pi$}. Here the notion of $\varepsilon$-far means that one has to change more than $\varepsilon dn$ edges to obtain a graph of maximum degree $d$ that has property $\Pi$. If $G$ is not $\varepsilon$-far from $\Pi$, then it is called \emph{$\varepsilon$-close}. To give a property tester on a bounded degree graph $G$, we assume that $G$ is given as an oracle, which allows us to perform \textit{neighbor queries} to $G$ such that for any input pair $(v,i)$, the oracle returns the $i$th neighbor of vertex $v$ if $i \le d_G(v)$, and a special symbol if $i > d_G(v)$, where $d_G(v)$ is the degree of $v$. This framework of graph property testing was initiated by Goldreich and Ron \cite{GR02:testing}. In this model, it is known that several properties are testable in constant time, such as hyperfinite properties \cite{NS13:hyperfinite} (see also \cite{CGRSSS14,GR02:testing} and the references therein). We now also know that properties such as bipartiteness \cite{GR98:sublinear} and expansion \cite{CS10:expansion,GR00:expansion,KS11:expansion,NS10:expansion} are testable in time $\widetilde{O}(\sqrt{n})$, with a nearly matching lower bound, and we need to perform at least $\Omega(n)$ queries to test $3$-colorability \cite{GR02:testing}. For more results, see recent surveys \cite{Gol11:testing,Ron10:testing}.

There are several ways to assess the cluster structure of a graph, such as $k$-means, cliques, modularity etc. One typically would want to argue that vertices in the same cluster should be well-connected and vertices from different clusters should be poorly-connected. In this paper, we use the concept of \textit{conductance} to measure the quality of the cluster structure of a graph. Given a graph $G = (V,E)$ with maximum degree bounded by $d$, and a subset $S\subseteq V$, the \emph{conductance of $S$} is defined as $\phi_G(S) := \frac{e(S, V \setminus S)}{d|S|}$, where $e(S, V \setminus S)$ denotes the number of edges coming out of $S$. Note that $\phi_G(V) = 0$. The \emph{conductance of the graph $G$}, denoted as $\phi(G)$, is defined as the minimum conductance value over all possible subsets $S$ of $V$ with $|S| \le |V|/2$. (For convenience, we define $\phi(G)=\frac{1}{d}$ if $G$ is the singleton graph, that is, the graph consisting of a single isolated vertex with no edges.) For any $S\subseteq V$, let $G[S]$ be the induced subgraph of $G$ on the vertex set $S$. Define the \textit{inner conductance} of $S$ to be the conductance of subgraph $G[S]$, namely, $\phi(G[S])$. To avoid confusion, we will also call the conductance $\phi_G(S)$ of $S$ in $G$ the \textit{outer conductance} of $S$.

Kannan et al.\ \cite{KVV04:clustering} introduced conductance as an appropriate measure of the quality of a cluster and this notion has been later used in numerous more applied works (see, e.g., \cite{Sch07:clustering}).
Further intuition has been employed to assert that a set $S$ with small outer conductance has few connections to the outside of $S$, and a graph $G$ with large conductance means that the vertices of $G$ are well-connected with each other. Following this intuition, Oveis Gharan and Trevisan \cite{OT14:expander} and Zhu et al.\ \cite{ZLM13:local} proposed to combine both outer conductance and inner conductance of a set $S$ to measure whether $S$ is a good cluster or not. That is, a set $S$ is considered to be a good cluster if $\phi_G(S)$ is small and $\phi(G[S])$ is large. In \cite{OT14:expander}, a graph $G$ is defined to be clusterable if $G$ can be partitioned into a number of disjoint parts so that each of them is such a good cluster. In this paper, we will use a related definition to characterize graphs with cluster structure.

%============================================================================

\subsection{Our results}
\label{subsec:result}

We begin with the formal definition characterizing graphs with a cluster structure and state our main results. The following definition is inspired by the work of Oveis Gharan and Trevisan \cite{OT14:expander}.

\begin{definition}
\label{def:clusterable-1}
For a $d$-degree bounded undirected graph $G = (V,E)$ with $n$ vertices and parameters $k,\phi,\varepsilon$, we define $G$ to be \emph{$(k,\phi)$-clusterable} if there exists a partition of $V$ into $h$ sets $C_1, \dots, C_h$ such that $1 \le h \le k$, and for each $i$, $1 \le i \le h$, $\phi(G[C_i]) \ge \phi$ and $\phi_G(C_i) \le c_{d,k} \varepsilon^4 \phi^2$, where for fixed $d,k$, $c_{d,k}$ is a universal constant. We call each $C_i$ a \emph{$\phi$-cluster} and the corresponding $h$-partition an \emph{$(h,\phi)$-clustering}.
\end{definition}

%\Artur{Since we're using this odd bound for $\phi_G(C_i) \le \frac{c_{d,k} \varepsilon^4 \phi^2}{\log n}$, one should explain it here.}
The above definition formalizes the idea that the existence of an edge is an indicator that two vertices are similar, i.e., two persons are friends or two authors belong to the same scientific community, while the lack of an edge is a (weaker) sign of the opposite statement. Therefore, a cluster should be, intuitively, well-connected in the inside and poorly-connected to the outside. (We remark that the gap between the conductance of $C_i$ and $G[C_i]$ in Definition \ref{def:clusterable-1} is a feature of our approach rather than an inherent property of the problem.)

In this paper, we develop an algorithm that with probability at least $\frac23$, accepts every $(k,\phi)$-clusterable graph and rejects every graph that is $\varepsilon$-far from every $(k,\phi^*)$-clusterable graph, where $\phi^*=O_{d,k}(\frac{\phi^2\varepsilon^4}{\log n})$. (Throughout the paper we use the notation $O_{d,k}()$ to describe a function in the Big-Oh notation assuming that $d$ and $k$ are constant.) %\Pan{Actually, our proofs do not require $k$ to be constant, but I am not sure if it is a good idea to point it out since this notation $O_{d,k}$ and $c_{d,k}$ are more acceptable if both $d,k$ are constants.}\Artur{If $k$ doesn't have to be a constant then we could add a note about it, e.g., in Conclusions (just saying so, without any further comments).}
Our main result is that we can distinguish such a clusterable graph from all graphs that are far from being clusterable
%with sublinear query and time complexity.
in sublinear time.
%\Artur{It doesn't make sense to have both running time and query complexity.}

\begin{theorem}
\label{thm:main}
Let $c'_{d,k}$ be a suitable constant depending on $d$ and $k$. There exists an algorithm that accepts every $(k, \phi)$-clusterable graph of maximum degree at most $d$ with probability at least $\frac23$, and rejects every graph of maximum degree at most $d$ that is $\varepsilon$-far from being $(k, \phi^*)$-clusterable with probability at least $\frac23$, if $\phi^* \le c'_{d,k} \frac{\phi^2 \varepsilon^4}{\log n}$. The running time of the algorithm is $\frac{\sqrt{n}}{\phi^2}(k \log n / \varepsilon)^{O(1)}$.

%Furthermore, if the graph is $(k,\phi)$-clusterable, then for each $\phi$-cluster of size at least $\varepsilon |V|$, at least one vertex from it will be found.  %the algorithm outputs a representative\Artur{This notion is vague and undefined. One should describe more details and be more precise.} for each $\phi$-cluster of more than $\varepsilon |V|$ vertices.
%\Artur{Does it make sense? I believe this is what we would like to say about our set of representatives, am I right? If so, then please check if it's OK that in the last word I'm writing about $(k,\phi)$-clustering rather than about $(h,\phi)$-clustering for $h \le k$.\\If you think a similar description is OK, then one would need to add a sentence or two in the proof of Theorem \ref{thm:main} describing how we will achieve this property.}
%
%Furthermore, if the graph is $(k,\phi)$-clusterable, then our property tester will return a randomly chosen set $S$ of $O(k \ln k \cdot \varepsilon^{-2})$ vertices that will provide a skeleton of the cluster structure: with probability at least $\frac23$, we can determine which vertices from $S$ belong to the same $\phi$-cluster in the $(k,\phi)$-clustering\Pan{We cannot draw such a conclusion since if $G$ is $k$-clusterable, the similarity graph might be just a clique, in which we cannot tell if two vertices belong to the same cluster or not. We may only guarantee that one representative of large clusters are selected.}.
\end{theorem}

One can question whether the gap between $\phi^*$ and $\phi$ in the form $\phi^* = O_{d,k}(\frac{\phi^2 \varepsilon^4}{\log n})$ or similar is really required. We believe that for an algorithm with a somewhat similar time complexity, both the $\log n$ and the $\varepsilon$ factors in the gap between $\phi$ and $\phi^*$ are necessary. For further discussion about this gap size we refer to Section \ref{subsection:Expansion}.

Note also that in our results we allow for clusterings with at most $k$ clusters (rather than with exactly $k$ clusters). This can be justified by the fact that in the property testing framework, every $(k,\phi)$-clusterable graph with exactly $h \le k$ clusters is $\varepsilon$-close to some $(k,\phi^*)$-clusterable graph with exactly $k$ clusters, for any reasonable choice of parameters (one can simply remove all edges that are incident to $k-h$ vertices).

%============================================================================

\subsection{Comparison with testing expansion and discussion of the gap size}
\label{subsection:Expansion}

For $k=1$, our problem is equivalent to that of \emph{testing graph expansion}, the problem which has received significant attention in the past.
Goldreich and Ron \cite{GR00:expansion} were the first to study this problem in details and proved a lower bound $\Omega(\sqrt{n})$ on the number of queries for testing graph expansion in the bounded degree model. This result has been complemented by a proposed algorithm, which Goldreich and Ron conjectured to be a property tester for the second largest eigenvalue (denoted by $\eta_2$) of the normalized adjacency matrix of a regularized version of the graph, in the sense that it accepts every graph with $\eta_2 \le \eta$ and rejects every graph that is $\varepsilon$-far from having $\eta_2 \le \eta^{\Theta(\mu)}$ for any $\mu > 0$. Note that by Cheeger's inequality (cf. Theorem \ref{thm:cheeger}), resolving of this conjecture would imply that the algorithm is also a property tester
%for expansion such that it accepts any graph with $\phi(G) \ge \sqrt{2(1-\eta)}$ and rejects every graph that is $\varepsilon$-far from having expansion at least $\frac{1-\eta^{\Theta(\mu)}} {2}$; or equivalently, the algorithm
that accepts any graph with $\phi(G) \ge \phi$ and rejects every graph that is $\varepsilon$-far from being a $\phi^*$-expander for $\phi^* = O(\mu \phi^2)$, where a graph $G$ is called a $\phi$-expander if $\phi(G)\ge \phi$.
%having expansion at least $\frac{1-(1-\frac{\phi^2} {2})^{\Theta(\mu)}}{2} = O(\mu \phi^2)$.
%\Artur{I reformulated that sentence since the old one was confusing and could be understood that Goldreich and Ron \cite{GR00:expansion} already gave a property tester and gave its analysis. Please check if the parameter $\mu$ is used correctly.}
%For this problem Goldreich and Ron \cite{GR00:expansion} proposed an algorithm with query complexity $\widetilde{O}(n^{0.5+\mu})$ (for any $\mu>0$) and a nearly matching lower bound $\Omega(\sqrt{n})$ for testing graph expansion in the bounded degree model. They conjectured that the algorithm is a property tester for the second largest eigenvalue (denoted by $\eta_2$) of the normalized adjacency matrix of a regularized version of the graph in the sense that it accepts every graph with $\eta_2\ge \eta$ and rejects every graph that is $\varepsilon$-far from having $\eta_2\ge \eta^{\Theta(\mu)}$.
Czumaj and Sohler \cite{CS10:expansion} proved a weaker version of this conjecture by showing that the algorithm from \cite{GR00:expansion} can distinguish in time $\widetilde{O}(\sqrt{n})$ any $\phi$-expander graph from graphs that are $\varepsilon$-far from being a $\phi^*$-expander for $\phi^* = O(\frac{\phi^2}{\log n})$. Kale and Seshadhri \cite{KS11:expansion} and Nachmias and Shapira \cite{NS10:expansion} extended this result and proved that in $\widetilde{O}(n^{0.5+\mu})$ time the algorithm accepts graphs with expansion $\phi$ and reject graphs which are $\varepsilon$-far from having expansion $\phi^*= O(\mu\phi^2)$.
%\Artur{I'm a little uneasy about this paragraph. Why do we have $\eta_2$ in the reference to \cite{GR00:expansion} when we don't use $\eta_2$ later, when referring to newer results? This is confusing. {\bf Pan:} I added two sentences (in bold) to explain the relation of $\eta_2$ and expansion. It looks a bit verbose but I am not sure how to simplify.}

Since the best known methods require a gap between $\phi$ and $\phi^*$ already for the special case $k=1$, it is clear that our work will also need a similar gap. It seems to be tempting to conjecture that --- similarly to the case of testing expansion --- it will suffice to reject (in the soundness) graphs that are $\varepsilon$-far from being $(k,\Theta(\mu\phi^2))$-clusterable for any $\mu>0$, instead of having a $\log n$ factor dependency between $\phi$ and $\phi^*$, as in our result. However, we do not think that this is possible and in the following we briefly sketch the differences from testing expansion and argue why the approach that led to a better gap for testing expansion is likely to fail (of course, this does not rule out other approaches, but this points to substantial obstacles to obtain an improved result).

Let $u,v$ be any two vertices in the graph $G$, which for simplicity is now assumed to be $d$-regular and connected. Let $\lambda_i$ be the $i$-th smallest eigenvalue and $\vv_i$ be the corresponding eigenvector of the (normalized) Laplacian of $G$. It is known that the lazy random walk on $G$ converges to the uniform distribution on its end-vertex. One can write (cf. Section \ref{sec:proofs} for details) the $l_2^2$-distance between the distribution $\p_v^\ell$ and $\p_u^\ell$ of the endpoints of the lazy random walks on $G$ of length $\ell$ starting at $v$ and $u$, respectively, as
\begin{displaymath}
    \norm{\p_v^\ell - \p_u^\ell}_2^2
        =
    \sum_{i=1}^n(\vv_i(u)-\vv_i(v))^2(1-\frac{\lambda_i}{2})^{2\ell}
    \enspace.
\end{displaymath}
Since a lazy random walk on a regular graph converges to the uniform distribution, we have $\vv_1(u) = \vv_1(v) = 1/\sqrt{n}$. Therefore, in the case $k=1$, by the fact that $0 = \lambda_1 \le \lambda_2 \le \dots \le \lambda_n \le 1$, we can upper bound $\norm{\p_v^\ell - \p_u^\ell}_2^2$ by bounding the second smallest eigenvalue and by making a proper choice of the length of the walk $\ell$.

If we want to extend this approach to $k > 1$, then our definition implies (cf. \cite{LOT12:high}) that in a $(k,\phi)$-clusterable graph there is a significant gap between $\lambda_{h}$ and $\lambda_{h+1}$ for some $h$, $1  \le h \le k$, where $h$ corresponds to the number of clusters in the instance. Now, assume for simplicity that $h=k$. Then we obtain that
\begin{displaymath}
    \norm{\p_v^\ell - \p_u^\ell}_2^2
        =
    \sum_{i=1}^k(\vv_i(u)-\vv_i(v))^2(1-\frac{\lambda_i}{2})^{2\ell}
        + \sum_{i=k+1}^n(\vv_i(u)-\vv_i(v))^2(1-\frac{\lambda_i}{2})^{2\ell}
    \enspace.
\end{displaymath}
We can upper bound $\sum_{i=k+1}^n(\vv_i(u) - \vv_i(v))^2 (1-\frac{\lambda_i}{2})^{2\ell}$ by using the bound for $\lambda_{h+1}$ in a similar way we can bound the entire term by bounding $\lambda_2$ in the case $k=1$. However, the critical part is the first summand. It turns out that there are instances where the average $l_2^2$-distance between $u, v$ from the same cluster is $\Omega(\frac{\phi^*}{d^3n})$ for a certain reasonable choice of $\ell$, such that the random walk mixes well in the cluster while does not escape from some non-expanding set containing the cluster too often
%\Pan{new}
%that allows the second term to be sufficiently small
% OK, let is be so
%\Artur{I'm not sure if I understand this, it sounds ambiguous. Any chance to rephrase it and make it more clear? (Think: we are saying the distance between $u,v$ from the same cluster is $\Omega(\frac{\phi^*}{d^3n})$ and at the same time the second term is sufficiently small; this doesn't make sense as a conclusion.) {\bf Pan:} How about now?}
(for more details, see discussions below and Appendix \ref{subsec:evidence-tight}). This seems to rule out an approach similar to \cite{KS11:expansion,NS10:expansion}, as this approach requires a significantly smaller distance between $\p_v^\ell$ and~$\p_u^\ell$.

%============================================================================

\subsection{Our techniques}

We develop the first sublinear algorithm for testing if a graph is $(k,\phi)$-clusterable, significantly extending earlier works on testing the expansion of a graph. Our algorithm draws a random sample set and tests for every pair of sample vertices if the distributions of the endpoints of a random walk starting at the two vertices are close in the $l_2^2$-distance. If this is the case, then it connects the two sample vertices by an edge in a \emph{similarity graph}. At the end, the algorithm accepts the input graph if the similarity graph is a collection of at most $k$ connected components. %Each such clique\Artur{Shouldn't this be a \emph{connected component} rather than a clique? In fact, it seems that this notion is inconsistently used throughout the paper and one should fix it. I, however, refrain from doing so because may, in some places, you would prefer to keep it as a clique, I'm not 100\% sure. I would suggest that someone does a search of all uses of word ``clique'' and check if they're as they should be, or not.} corresponds to one big cluster and the vertices of the clique can be viewed as its representatives.

Our main new contributions are as follows.
\begin{itemize}
\item Our algorithm  is the first property tester that directly makes use of testing pairwise closeness of distributions induced by random walks. Previous related algorithms \cite{CS10:expansion,GR00:expansion,KPS13,KS11:expansion,NS10:expansion} tested if the distribution of the endpoints a random walk starting at a vertex $v$ is close to the uniform distribution and then drew their conclusions about the structure of the graph. In our case, we do not know how the distribution looks like (it will be close to uniform \emph{inside} every cluster, but this is not very helpful since the cluster is unknown to us and the support size of a distribution is hard to estimate \cite{RRSS09:support}) and it may have significant distance from the uniform distribution.
\item It is the first property tester that exploits (in the completeness case) a ``somewhat stable'' behaviour of the random walk distribution at a length where it is significantly different from the stationary distribution, i.e., we pick the length of the random walk in such a way that it is almost stable on its own cluster, and most of the probability mass will stay in some non-expanding set containing the cluster.%\Pan{new}.
%\footnote{\bf CS: Is the (half-)sentence starting with 'but' still true under the new algorithm?}
%\item We incorporate new tools from spectral graph theory, like a higher order Cheeger's inequality \cite{LOT12:high}, to analyze the completeness of the tester.\Artur{This ``main new contribution'' is lame. It's an interesting feature of our paper/algorithm, but I wouldn't be selling it as one of three ``main contributions.'' I would suggest to remove this item, or rephrase it, and not call it as one of ``main contributions.''}
\end{itemize}

In order to test closeness of distribution, we use a recent tester for closeness of distributions in $l_2$-norm by Chan et al.\ \cite{CDVV14:optimal}, which gives slightly better bounds than the corresponding tester of Batu et al.\ \cite{BFRSW13:testdist}. A combination with a necessary condition on the $l_2$-norm of the distribution of the endpoints of the random walk from the sample vertices leads to improved bounds. It is tempting to think of this problem in the setting of $l_1$-norm since, for example, the distance between a random walk starting from different clusters is typically $\Omega(1)$ in $l_1$-norm. But this is misleading. It is known that no \emph{stable} $l_1$-tester exists, i.e., $l_1$-testers cannot distinguish the case that distributions are close from the case that they are not \cite{VV11:linear} ($l_1$-testers can only distinguish between identical (or almost identical) distributions and distributions that are far away from each other). However, as already explained in the previous section, we cannot hope for distributions to be arbitrarily close even if the random walks start in the same cluster.
%This implies that $l_1$-tester cannot distinguish whether two distributions are from the same cluster or two different clusters,
To address this difficulty, we will use the fact (noted earlier by Batu et al. \cite{BFRSW13:testdist}) that an $l_1$-tester can be reduced to an $l_2$-tester if the probability of every item is $O(n^{-1})$, which is likely to be the case if the graph is $(k,\phi)$-clusterable.
%On the other hand, as already noted by Batu et al., $l_1$-tester\ can be reduced to an $l_2$-tester if the probability of every item is $O(n^{-1})$, which is likely to be the case in our scenario.

We note that in the $l_2^2$-distance, a typical distance between the distribution of the endpoints of the random walks starting in two vertices from different clusters can be very small. For example, if we have two disconnected expanders (clusters) on $n/2$ vertices each, then for a sufficiently long random walk the distribution of the endpoints of the walk will be (almost) uniform on the cluster of the starting vertex. Therefore the distance between the distributions of the endpoints of random walks starting in different clusters will be $O(1/n)$. Furthermore, as we have argued above, the distance between the distributions of the endpoints of random walks will not be much smaller in the case that they come from the same cluster. Analyzing these two cases is one of the central technical challenges of our paper.

%============================================================================

\subsection{Other related work}

In the context of property testing, Alon et al.\ \cite{ADPR03:clustering} studied the problem of testing if a set of points in $\mathbb{R}^d$ is clusterable (see also \cite{CS05}), but both their problem definition and techniques are quite different from ours. Kale et al.\ \cite{KPS13} gave a sublinear expansion reconstruction algorithm that outputs the neighborhood of any input vertex $v$ in a $\Omega(\frac{\phi^2}{\log n})$-expander $G'$ that is $\frac{\phi \varepsilon}{\log n}$-close to the input graph $G$, which is assumed to be $\varepsilon$-close to a $\phi$-expander. In particular, they designed an algorithm that runs in $\widetilde{O}(\sqrt{n})$-time and distinguishes vertices from a large set that induces an expander from vertices that belong to a bad cut, by using uniform averaging random walks and testing if the distribution of endpoints of the walk is close to uniform distribution (in the $l_1$-norm distance) or not. This work does not (directly) compare distributions of the endpoints of the random walks starting from different vertices, as we do in our paper.

Our work is closely related to works on testing distributions. Batu et al.\  \cite{BFRSW00:distribution,BFRSW13:testdist} were the first to give sublinear time algorithms for testing the closeness of two discrete distributions and since then, a large body of work has been devoted to the problem of estimating the properties of distributions from a small number of samples (see the recent survey \cite{Rub12:distribution} and the reference therein). In particular, Levi et al.\ \cite{LRR13:multidistri} gave an algorithm with complexity $\widetilde{O}(n^{2/3})$ to test whether a set of distributions over a domain of size $n$ can be partitioned into $k$ clusters. Very recently, Chan et al.\  \cite{CDVV14:optimal} gave asymptotically optimal testers for the closeness of two distributions under both $l_1$ and $l_2$ settings.

Besides the related works in the literature of property testing, our work is also closely related to the area of graph partitioning and spectral clustering. Ng et al.\ \cite{NJW01:spectral} and Shi et al.\ \cite{SM00:spectral} used the first few eigenvectors of some matrices to partition a graph (or a set of data) into sparsely connected clusters. Different ways of measuring clustering based on intra-cluster density vs. inter-cluster sparsity and some experimental results were given  in \cite{BGW07}. Kannan et al.\ \cite{KVV04:clustering} proposed a bicriteria to measure the quality of a clustering, in which a good clustering is defined to be a partition of vertex sets such that each set in the partition has large inner conductance and few edges lying between different sets. They gave spectra based approximation algorithm for finding such a clustering. Lee et al.\ \cite{LOT12:high} and Louis et al.\ \cite{LRTV12:sparse} recently gave theoretical analysis of some spectral algorithms that use the first $k$ eigenvectors of the normalized Laplacian matrix for finding a $k$-partition of a graph such that each part is of small (outer) conductance, without any restriction on the inner conductance of the cluster. Zhu et al.\ \cite{ZLM13:local,OZ14:clustering} gave personal PageRank based and flow based local algorithms for finding a set of large inner conductance and small outer conductance. Makarychev et al.\ \cite{MMV12:semirand} studied a semidefinite programming based algorithm in the semi-random model to find such a set. Tanaka \cite{Tan13:partition} and Oveis Gharan and Trevisan \cite{OT14:expander} recently studied the existence and construction of a $k$-clustering such that each cluster is of large inner conductance and of small outer conductance, under the assumption that there is some gap between $\rho_G(k)$ and $\rho_G(k+1)$, where $\rho_G(k)$ is the minimum conductance of any $k$ disjoint subsets of the graph (cf. Section \ref{subsec:spectral-clusterable}). Dey et al.\ \cite{DRS14:spectral} considered the performance of a spectral clustering algorithm that applies a greedy algorithm for $k$-centers on some embedding induced by the first $k$ eigenvectors of the graph Laplacian. Peng et al.\ \cite{PSZ14:partitioning} studied the eigenvector structures of the Laplacian of well-clustered graphs (which is very related to our definition of clusterable graphs) and the approximation ratio of $k$-means clustering algorithms on these graphs.

%============================================================================

\subsection{Organization of the paper}
In Section \ref{sec:pre}, we give notations and definitions used throughout the paper. In Section \ref{sec:algorithm}, we give a formal description of our tester for clusterable graphs. We then present in Section \ref{sec:analysis} some central properties, which we use for proving our main result --- Theorem \ref{thm:main}. The proofs of these central properties are given in Section \ref{sec:proof-section5}. Section \ref{sec:conclusion} has final conclusions.
%
%In Appendix \ref{subsec:defn}, we present some basic tools from spectral graph theory. In Appendix \ref{subsec:evidence-tight}, we give further discussions on spectral properties of clusterable graphs and present intuitions why Lemma \ref{lem:clusterable-eigenvector} is essentially tight.
%
Finally, in Appendix we will present some auxiliary tools used in the analysis.

%============================================================================

\section{Preliminaries}
\label{sec:pre}

Let $G=(V,E)$ be an undirected and unweighted graph with maximum degree bounded by a constant $d$. Let $n:=|V|$. For a vertex $v\in V$, let $d_G(v)$ be the degree of $v$. We assume that $G$ is represented by its \textit{adjacency list} and that we can access $G$ through an oracle, which allows us to perform the \textit{neighbor query} to $G$. That is, when the oracle is given as input a vertex $v$ and an integer $i$, it outputs the $i$-th neighbor of $v$ if $d_G(v)\ge i$, and a special symbol otherwise (in constant time).

%\Artur{This (conductance, $\phi_G(S)$, $\phi(G)$, \dots) has been defined in Introduction. Why do have it twice?}
%For any two disjoint vertex sets $S, T \subseteq V$, let $e(S,T)$ denote the number of edges connecting $S$ and $T$. The \textit{conductance of $S$} in $G$ is defined as $\phi_G(S):=\frac{e(S,V\setminus S)}{d|S|}$. Note that $\phi_G(V)=0$. The \emph{conductance of the graph $G$} is defined to be the minimum conductance value of any possible subset of size no large than $n/2$, that is, $\phi(G):=\min_{S\subseteq V, |S|\le n/2}\phi_G(S)$. (For convenience, we define $\phi(G)=\frac{1}{d}$ if $G$ is the singleton graph.) For any $S\subseteq V$, let $G[S]$ be the induced subgraph of $G$ on the vertex set $S$. Define the \textit{inner conductance} of $S$ to be the conductance of subgraph $G[S]$, namely, $\phi(G[S])$. To avoid confusion, we will also call the conductance $\phi_G(S)$ of $S$ in $G$ the \textit{outer conductance} of $S$.

As mentioned in the introduction, we will use Definition \ref{def:clusterable-1} of $(k,\phi)$-clusterable graphs and $\phi$-clusters inspired by \cite{OT14:expander} to characterize the cluster structure of graphs and the clusters therein. Note that a $(1,\phi)$-clusterable graph is an expander graph with conductance $\phi$, which we abbreviate as $\phi$-expander (this should not be confused with $\phi$-cluster). %and $\alpha$-vertex-expander (in the proof of Lemma~\ref{lemma:partition-eps-far}).}.
%exactly the $(k,\phi,\frac{c_{d,k}\varepsilon^4 \phi^2}{\log n})$-clusterable graph in Definition~\ref{def:clusterable-old}.

We are interested in testing if a given graph is $(k,\phi)$-clusterable in sublinear time %\Artur{What is the difference between query and time complexity? If there is a real difference then it's not described in the paper. What about sticking to the runtime only? (Yes, some people like to refer to it as query complexity, some other people prefer running time, but keeping both terms doesn't make sense.)}
in the framework of property testing. Formally speaking, we will study the following problem: given parameters $k,\phi,\varepsilon$, and a $d$-degree bounded graph $G$, we want to test if $G$ is $(k,\phi)$-clusterable or $\varepsilon$-far from being $(k,\phi^*)$-clusterable with as few queries as possible, for $\phi^*$ being as close to $\phi$ as possible. We have the following definition of graphs that are $\varepsilon$-far from clusterable graphs.

\begin{definition}
~\label{def:eps-far}
A graph $G$ (of maximum degree at most $d$) is \emph{$\varepsilon$-far from $(k, \phi)$-clusterable} if we have to add or delete more than $\varepsilon d n$ edges to obtain a $(k, \phi)$-clusterable graph of maximum degree at most $d$. %(Here, we implicitly assume that the modified graph is on the same vertex set as $G$.)
If $G$ is not $\varepsilon$-far from $(k, \phi)$-clusterable then it is \emph{$\varepsilon$-close} to $(k, \phi)$-clusterable.
\end{definition}

%=============================================================================

\section{The algorithm}
\label{sec:algorithm}

In this section, we describe our algorithm used in Theorem \ref{thm:main}. We first introduce the following \textit{random walk} on a $d$-bounded degree graph $G$ that will be used in our algorithm. In this walk, if we are currently at vertex $v$, then in the next step, we choose randomly an incident edge $(v,u)$ with probability $\frac{1}{2d}$ and move to $u$. With the remaining probability, which is at least $\frac12$, we stay at $v$. Note that if we let $G_\reg$ denote the weighted $d$-regular graph that is obtained from $G$ by adding an appropriate number of half-weighted self-loops, then this random walk is exactly a \textit{lazy random walk} on $G_\reg$. We will let $\p_v^\ell$ denote the distribution of endpoints of such a random walk of length $\ell$ starting at $v$.
%As mentioned before, given as input a $d$-bounded degree graph $G$, the algorithm will use standard lazy random walks on the corresponding virtually constructed weighted $d$-regular graph $G_\reg$. %that is obtained by adding a suitable number of self-loops on each vertex in $G$ (cf. Section~\ref{sec:pre}).%\Pan{I changed the algorithm a bit to make it more formal}.
%\Artur{I still find this notion of ``virtually constructed weighted $d$-regular graph $G_\reg$'' vague and not defined as good as it could \dots but I'm not sure how to do it better. {\bf Pan:} I moved this paragraph from Section 2, and hope it helps.}
Our testing algorithm is given as follows.

\begin{center}
\begin{tabular}{|p{0.9\textwidth}|}
\hline
\textbf{$k$-Cluster-Test} $(G,s,\ell,\sigma,k)$\\
\hline
\begin{enumerate}
\item Sample a set $S$ of $s$ vertices independently and uniformly at random from~$V$.
\item For any $v\in S$, let $\p_v^\ell$ be the distribution of endpoints of random walk of length $\ell$ starting at $v$.
\item\label{alg:l2norm} For any $v\in S$, test if $||\p_v^\ell||_2^2>\sigma$; if so, then abort and reject.
\item\label{alg:distribution-closeness} For each pair $u, v \in S$: if $l_2$ distribution tester accepts that $\norm{\p_u^\ell-\p_v^\ell}_2^2\le \frac{1}{4n}$, then add an edge $(u,v)$ in ``similarity graph'' $H$ on vertex set~$S$.
\item If $H$ is the union of at most $k$ connected components, then accept; otherwise, reject.
\end{enumerate}\\
\hline
\end{tabular}
\end{center}

If the graph is $(k,\phi)$-clusterable then we will show that (for the right choice of parameters) the distributions of the endpoints of random walks will be close if they come from the same cluster.
%, and they will not be close if they come from different clusters.
Furthermore, Step~\ref{alg:l2norm} tests a necessary condition for the efficient $l_2$ distribution tester that will be used in~Step~\ref{alg:distribution-closeness}, i.e., $||\p_v^\ell||_2^2$ is small, which is satisfied for almost all vertices in a $(k,\phi)$-clusterable graph. The small $l_2^2$-norm property of distributions can then be exploited in the testing for closeness of distributions in Step \ref{alg:distribution-closeness} to obtain a better running time.

%=============================================================================

\subsection{Implementation of distribution testing}
\label{section:Distributions}

Our algorithm relies on an efficient tester for the $l_2$-closeness of two distributions $\p$ and $\q$. The tester used in Step \ref{alg:distribution-closeness} of \textbf{$k$-Cluster-Test} was recently proposed by Chan et al.\ \cite{CDVV14:optimal} and is similar to the $l_2$ distance tester in \cite{BFRSW13:testdist} that uses the statistics of collisions in the sample sets from both distributions $\p, \q$.
The following is a direct corollary of Theorem 1.2 from \cite{CDVV14:optimal}. %(Theorem 1.2 in \cite{CDVV14:optimal} gives the proof for $\delta = \frac34$ and the extension to arbitrary $\delta > 0$ is straightforward by performing independent repetitions.)

\begin{theorem}
\label{cor:distribution}
Let $c_{\ref{cor:distribution}}$ be some appropriate constant  $c_{\ref{cor:distribution}} \ge 1$. Let $\delta, \xi > 0$ and let $\p,\q$ be two distributions over a set of size $n$ with $b \ge \max\{\norm{\p}_2^2, \norm{\q}_2^2\}$. Let $r \ge c_{\ref{cor:distribution}} \cdot \frac{\sqrt{b}}{\xi} \ln\frac{1}{\delta}$. There exists an algorithm, denoted by \textbf{$l_2$-Distribution-Test}, that takes as input $r$ samples from each distribution $\p,\q$, and accepts the distributions if $\norm{\p-\q}_2^2\le \xi$, and rejects the distributions if $\norm{\p-\q}_2^2\ge 4\xi$, with probability at least $1-\delta$. The running time of the tester is linear in its sample size.
\end{theorem}

We also need an efficient algorithm to estimate the $l_2^2$-norm of the probability distribution of the endpoints of a random walk in a graph.
%(also called the collision probability).
In Step \ref{alg:l2norm} of our algorithm \textbf{$k$-Cluster-Test} we will use \textbf{$l_2^2$-norm tester}, the performance of which is guaranteed in the following lemma (the proof follows almost directly from the proof of Lemma 4.2 in \cite{CS10:expansion} that in turn is built on Lemma 1 in \cite{GR00:expansion}, cf. Appendix \ref{subsec:proof-on-distribution} for details).

\begin{lemma}
\label{lem:collision}
Let $G=(V,E)$ with $|V|=n$. Let $v \in V$, $\sigma>0$ and $r \ge 16\sqrt{n}$. Let $t\ge 1$ and let $\p_v^t$ be the probability distribution of the endpoints of a random walk of length $t$ from $v$. There exists an algorithm, denoted by \textbf{$l_2^2$-norm tester}, that takes as input $r$ samples from $\p_v^t$ and accepts the distribution if $\norm{\p_v^t}_2^2\le \sigma/4$ and rejects the distribution if $\norm{\p_v^t}_2^2>\sigma$, with probability at least $1-\frac{16\sqrt{n}}{r}$. The running time of the tester is linear in its sample size.
\end{lemma}

%=============================================================================

\section{Analysis of $k$-Cluster-Test}
\label{sec:analysis}

We outline the proof of our main theorem, Theorem \ref{thm:main}. Our techniques are based on two intuitions. The first intuition is that if two ``typical'' vertices $u, v$ are \emph{from the same large cluster}, then the distributions of the endpoints of two \emph{sufficiently long} random walks starting at $u, v$, respectively, are close; and if $u, v$ are \emph{separated by a non-expanding cut}, then the distributions of the endpoints of two \emph{not so long} random walks from $u, v$, respectively, are far away from each other. If this intuition holds, then we can reduce our problem to the problem of testing the closeness of two distributions, and then use the returned results to decide whether the distributions induced by the random walks from different sampled vertices can be divided into $k$ groups or not. In particular, if our input graph $G$ is $(k,\phi)$-clusterable, then we can get at most $k$ connected components in our ``similarity graph'' $H$. (Actually, as will be seen from our proof, sampled vertices from the same cluster form a clique in $H$.) % \Artur{Or a connected component?}.
On the other hand, if $G$ is far from being $(k,\phi^*)$-clusterable, then we expect that we can get at least $k+1$ \junk{cliques or non-clique} connected components in $H$ . The latter is based on our second intuition that if $G$ is far from being $(k, \phi^*)$-clusterable, then there are at least $k+1$ (large) well separated sparse cuts. We present several lemmas that formalize these intuitions in Section \ref{subsec:keyproperties} and then give the proof of Theorem \ref{thm:main} in Section \ref{sec:proof-main-theorem}.

%=============================================================================

\subsection{Key properties}
\label{subsec:keyproperties}

In this section, we state several lemmas describing the properties used in our analysis of \textbf{$k$-Cluster-Test}. The proofs of the results are deferred to Section \ref{sec:proof-section5}.

In the following we will formally state these key properties under the definition of a more general class of clusterable graphs, even though our main focus is on the study of properties of $(k,\phi)$-clusterable graphs. To study detailed properties of $(k,\phi)$-clusterable graphs and their dependencies on all parameters, we will use the following, more general definition of $(k,\phi_{in},\phi_{out})$-clusterable graphs, which follows the framework from \cite{OT14:expander}.

\begin{definition}
\label{def:clusterable-old}
For an undirected graph $G$, and parameters $k,\phi_{in},\phi_{out}$, we define $G$ to be \emph{$(k,\phi_{in},\phi_{out})$-clusterable} if there exists a partition of $V$ into $h$ subsets $C_1, \dots, C_h$ such that $1 \le h \le k$ and for each $i$, $1 \le i \le h$, $\phi(G[C_i])\ge \phi_{in}$, $\phi_G(C_i)\le \phi_{out}$. We call each $C_i$ a \emph{$(\phi_{in},\phi_{out})$-cluster} and the corresponding $h$-partition an \emph{$(h,\phi_{in},\phi_{out})$-clustering}.
\end{definition}

We can define a graph $G$ to be $\varepsilon$-far from $(k,\phi_{in},\phi_{out})$-clusterable similarly to Definition~\ref{def:eps-far}. Note that a $(k,\phi)$-clusterable graph from Definition \ref{def:clusterable-1} is exactly a $(k,\phi, c_{d,k} \varepsilon^4 \phi^2)$-clusterable graph from Definition \ref{def:clusterable-old}.

We first show that if the graph is $(k,\phi_{in},\phi_{out})$-clusterable then for any large cluster $C$ with $\phi(G[C])\ge \phi_{in}$, there exists a large subgraph $\widetilde{C}$ such that the distributions of the endpoints of two random walks of length large enough starting from any two vertices $u,v\in \widetilde{C}$ are close in the $l_2$-norm (that is, the $l_2$ distance between $\p_u^\ell$ and $\p_v^\ell$ is small). The proof of this result relies on spectral properties of clusterable graphs given in  Section \ref{subsec:spectral-clusterable}. %Let $c$ be a universal constant that will be specified later. %Recall that we required the outer conductance of $\phi$-cluster is at most $c_{d,k}\varepsilon^4\phi^2/\log n$.
%, where $c_{d,k}=\frac{c}{d^2k^5\ln^2 k}$ for a sufficiently small universal constant $c>0$.

\begin{lemma}
\label{lem:smalll2}
Let $0 < \alpha, \beta < \frac12$. If $G = (V,E)$ is $(k, \phi_{in}, \phi_{out})$-clusterable, and $C \subseteq V$ is any subset such that $|C| \ge \beta n$ and $\phi(G[C]) \ge \phi_{in}$, then there exists $\alpha_{\ref{lem:smalll2}} = \alpha_{\ref{lem:smalll2}}(k, \alpha, \beta,d)$ and a universal constant $c_{\ref{lem:smalll2}}>0$ such that for any $t \ge \frac{c_{\ref{lem:smalll2}}k^4 \log n}{\phi_{in}^2}$, $\phi_{out}\le \alpha_{\ref{lem:smalll2}}\phi_{in}^2$,  %\frac{2c_{\ref{thm:highcheeger}}^2k^4}{\phi_{in}^2}\ln(4n)$,
there exists a subset $\widetilde{C}\subseteq C$ with $|\widetilde{C}| \ge (1-\alpha)|C|$ such that for any $u, v \in \widetilde{C}$, the following holds:
\begin{displaymath}
    \norm{\p_u^t-\p_v^t}_2^2\le \frac{1}{4n}
    \enspace.
\end{displaymath}
\end{lemma}

\junk{\red{(Remove this paragraph?) We remark that in the proof of Theorem~\ref{thm:main}, we will choose appropriate parameters $\alpha,\beta$ such that for any $\phi$-cluster $C$, if we let $\phi_{in}=\phi(G[C])=\phi$ and $\phi_{out}=\phi_G(C)\le c_{d,k}\varepsilon^4\phi^2$, then it always holds that $\phi_{out}\le \alpha_{\ref{lem:smalll2}}\phi_{in}^2$ by our definition of $c_{d,k}$.}}

In order to use an efficient distribution tester (e.g., as the one given in Theorem \ref{cor:distribution}), we need to guarantee that for a large fraction of vertices a sufficiently long random walk starting from a typical vertex will induce a distribution of its endpoints with small $l_2$-norms. We will prove the following lemma using spectral analysis of clusterable graphs.

\begin{lemma}
\label{lem:smallhnorm}
Let $0 < \alpha < 1$. If $G$ is $(k, \phi_{in}, \phi_{out})$-clusterable, then there exists $V' \subseteq V$ with $|V'|\ge (1-\alpha)|V|$ such that for any $u \in V'$ and any $t \ge \frac{c_{\ref{lem:smallhnorm}} k^4 \log n}{\phi_{in}^2}$, for some universal constant $c_{\ref{lem:smallhnorm}}>0$,
the following holds:
\begin{displaymath}
    \norm{\p_u^t}_2^2\le \frac{2k}{\alpha n}
    \enspace.
\end{displaymath}
\end{lemma}

Note that the above lemma does not require any assumption about $\phi_{out}$, and thus applies directly to any $(k,\phi)$-clusterable graphs by substituting $\phi$ for $\phi_{in}$ in the lemma.

For the soundness of our algorithm, we need the following lemma that shows that given two well separated sets $A, B \subseteq V$, for any two ``typical'' vertices $u \in A$, $v \in B$, the $l_2$-norm of the difference between the corresponding distributions of endpoints of random walks of short length starting from $u, v$ will be large. Our proof relies on the fact that any %\Artur{What do you mean by ``sparse set $A$''?}
set $A$ with small outer conductance has a large subset $\widehat{A}$ such that the random walk starting from any vertex in $\widehat{A}$ will stay inside $A$ for a relatively long time.%(see, e.g., \cite{ST13:local}).

\begin{lemma}
\label{lem:largel2}
Let $\alpha$ and $\psi$ be arbitrary with $0 < \alpha, \psi < 1$. Let $A \subseteq V$ be any subset of $G$ such that $\phi_G(A) \le \psi$. Then for any $t \ge 1$, there exists a subset $\widehat{A} \subseteq A$ with $|\widehat{A}| \ge (1-\alpha)|A|$ such that for any $v \in \widehat{A}$, the probability that the random walk of length $t$ starting from vertex $v$ never leaves $A$ in all $t$ steps is at least $1 - \frac{t \psi}{2 \alpha}$.

Furthermore, for any $t$, $1 \le t \le \frac{\alpha}{2 \psi}$, any two disjoint subsets $A, B \subseteq V$ with $\phi_G(A), \phi_G(B) \le \psi$, and any two vertices $u, v$ such that $u \in \widehat{A}, v \in \widehat{B}$, the following holds:
\begin{displaymath}
    \norm{\p_u^t - \p_v^t}_2^2
        \ge
    \frac{1}{n}
    \enspace.
\end{displaymath}
\end{lemma}

\begin{remark}
We note that the above lower bound %\footnote{\bf CS: If this refers to the previous lemma (which I think it does), then this should be 'lower bound')}
is almost tight up to constants. Consider the graph that is composed of two disconnected parts such that each of them is a $\phi_{in}$-expanders of size $n/2$. Then for any two starting vertices $u,v$ from two different parts, for $t=\Theta(\frac{\log n}{\phi_{in}^2})$, both $\p_u^t$ and $\p_v^t$ will be very close to the uniform distribution on each cluster, and therefore, the $l_2^2$ distance between these two distributions will be $O(1/n)$. %$\norm{\p_u^t-\p_v^t}_2^2=\norm{\p_u^t}_2^2+\norm{\p_v^t}_2^2\ge \frac{\norm{\p_u^t}_1+\norm{\p_u^t}_1}{n/2}=\frac{4}{n},$ where the inequality follows from the Cauchy-Schwarz inequality and the fact that the support of each distribution is at most $n/2$. %\Artur{To be done. I guess, this is not the best place for this remark; it should be somewhere earlier.
\end{remark}

For the analysis showing that graphs far from clusterable will be rejected, we will use a property that if a graph $G=(V,E)$ is $\varepsilon$-far from any $(k, \phi_{in}^*,\phi_{out}^*)$-clusterable graph, then its vertex set $V$ can be partitioned into $k+1$ subsets $V_1, \dots, V_{k+1}$, each of linear size and of small outer conductance.

\begin{lemma}
\label{lemma:partition-eps-far-improved}
Let $\alpha_{\ref{lemma:partition-eps-far-improved}} = \alpha_{\ref{lemma:partition-eps-far-improved}}(d,k)$ be a certain constant that depends on $d$ and $k$. If $G = (V,E)$ is $\varepsilon$-far from $(k, \phi^*_{in}, \phi^*_{out})$-clusterable with $\phi^*_{in} \le \alpha_{\ref{lemma:partition-eps-far-improved}} \cdot \varepsilon$, then there exist a partition of $V$ into $k+1$ subsets $V_1, \dots, V_{k+1}$ such that for each $i$, $1 \le i \le k+1$, $|V_i| \ge \frac{1}{1152k} \varepsilon^2 |V|$ %\Artur{I put here the new constant (1152 instead of 288) which comes from our analysis in Section \ref{sec:epsfarlemma}. Please check if this doesn't impact any further calculations, estimations, etc.}
and $\phi_G(V_i) \le c_{\ref{lemma:partition-eps-far-improved}} \phi^*_{in} \varepsilon^{-2}$, for some constant $c_{\ref{lemma:partition-eps-far-improved}} = c_{\ref{lemma:partition-eps-far-improved}}(d,k)$ and for any $0 \le \phi^*_{out} \le 1$.
\end{lemma}

%ARTUR: I'm not sure if these comments should stay and
%   so as for now, I deleted it:
\junk{
Let us remark here that in Lemma \ref{lemma:partition-eps-far} we do not have any explicit bound for $\phi^*_{out}$ and any $0 \le \phi^*_{out}\le 1$ will do. This is the case under our assumption that $\phi^*_{in} \le O(\varepsilon)$. Actually, we can show that the above claim does not hold when $\phi^*_{in}$ is some large constant and $\phi^*_{out}$ is close to $0$. For example, let $G$ be a graph consisting of $k$ great expanders with $\phi(G[V_i]) =\Theta(1)$ and these $k$ parts are ``evenly'' connected by $\approx \varepsilon n$ edges, which is $\varepsilon$-far from $(k,\Theta(1),0)$-clusterable, while it is unlikely to partition it into $k+1$ large and sparse parts, as required in our proof for soundness.
}

%Our main result Theorem \ref{thm:main} follows from Lemma \ref{lem:smalll2}--Lemma \ref{lemma:partition-eps-far}. The details of the proof are given in Section \ref{sec:proof-main-theorem}.

%\begin{lemma}
%\label{lemma:partition-eps-far}
%Let $c_d$ be a certain constant that depends on $d$. If $G = (V,E)$ is $\varepsilon$-far from $(k, \phi^*_{in}, \phi^*_{out})$-clusterable with $\phi^*_{in} \le \frac{3}{c_d} (\varepsilon/24)^{k-1}$, then there is a partition of $V$ into $k+1$ sets $V_1, \dots, V_{k+1}$ such that for each $i$, $1 \le i \le k+1$, $|V_i| \ge \alpha_{\ref{lemma:partition-eps-far}} |V|$ and $\phi_G(V_i) \le c_{\ref{lemma:partition-eps-far}}\phi^*_{in}$, for some constants $\alpha_{\ref{lemma:partition-eps-far}}=\alpha_{\ref{lemma:partition-eps-far}}(\varepsilon,k)$,
%$c_{\ref{lemma:partition-eps-far}}=c_{\ref{lemma:partition-eps-far}}(c_d,\varepsilon,k)$
%\end{lemma}

%=============================================================================

\subsection{Proof of main result --- Theorem \ref{thm:main}}
\label{sec:proof-main-theorem}

We will use Lemmas \ref{lem:smalll2}--\ref{lemma:partition-eps-far-improved} to prove our main result --- Theorem \ref{thm:main}. In the rest of this section, we prove the completeness, soundness and analyze the running time of the tester \textbf{$k$-Cluster-Test}.

In the algorithm \textbf{$k$-Cluster-Test}, we set $s = \frac{1536 k \ln(8(k+1))}{\varepsilon^2}$, $\ell = \frac{\max\{c_{\ref{lem:smalll2}}, c_{\ref{lem:smallhnorm}}\}\cdot k^4 \log n}{\phi^2}$, $\sigma = \frac{192sk}{n}$. %, $\zeta=\frac{36 r^2 s k}{n} = O(\frac{k^6 (\ln k/\varepsilon)^6}{\varepsilon^8})$.
We set $r = 192 c_{\ref{cor:distribution}} s \sqrt{skn} \ln s = O(\frac{k^2 (\ln k/\varepsilon)^{5/2} \sqrt{n}}{\varepsilon^3}), b = \frac{216 s k}{n}$, $\xi = \frac{1}{4 n}$, $\delta = \frac{1}{12 s^2}$ in Theorem \ref{cor:distribution}, and set $r = 192 c_{\ref{cor:distribution}} s \sqrt{skn} \ln s$ and $\sigma=\frac{192sk}{n}$ in Lemma~\ref{lem:collision}.

We specify now the constant $c_{d,k}$ that we used in the definition of a $\phi$-cluster to be $c_{d,k} = \frac{\alpha_{\ref{lem:smalll2}}(k, \frac{1}{24s}, \frac{1}{24ks}, d)}{\varepsilon^4} =\frac{c}{k^5 d^4 \ln^2 (8(k+1))}$ for a universal constant $c$.%\Pan{Some of them need to be changed depending on the new description of the distribution testers.}
%\Artur{I'm not going to even try to check if these bounds and constants make sense. It's insane.}

%=============================================================================

\subsubsection{Completeness --- accepting $(k, \phi)$-clusterable graphs}

We begin with showing that the algorithm \textbf{$k$-Cluster-Test} will accept $k$-clusterable graphs.

\begin{lemma}
\label{lemma:completeness}
If the input graph $G$ is $(k, \phi)$-clusterable, then with probability at least $\frac23$, the algorithm \textbf{$k$-Cluster-Test} accepts $G$. %, and it outputs at least one vertex %\Artur{The use of ``representatives'' need to be revised. The representatives have been never properly defined and never properly explained.}
%from each $\phi$-cluster of size at least $\varepsilon |V|$.
\end{lemma}

\begin{proof}
As indicated in the algorithm, we consider random walks of length $\ell$. We apply Lemmas \ref{lem:smalll2} and \ref{lem:smallhnorm} to the $(k,\phi)$-clusterable graph $G$, and we set $\phi_{in} = \phi$, $\phi_{out} = c_{d,k} \varepsilon^4 \phi^2$, $t = \ell$, $\alpha = \frac{1}{24s}$, and $\beta = \frac{1}{24ks}$ in the lemmas. Note that by our definition of $\phi$-cluster, the outer conductance of the cluster is at most $c_{d,k}\varepsilon^4 \phi^2 \le \alpha_{\ref{lem:smalll2}} \phi^2$, since $c_{d,k} \varepsilon^4 = \alpha_{\ref{lem:smalll2}}(k, \frac{1}{24s}, \frac{1}{24ks}, d)$, which implies that the conditions of Lemma \ref{lem:smalll2} are satisfied for any $\phi$-cluster of size at least $\beta n$ in $G$. Since $\ell = \frac{\max\{c_{\ref{lem:smalll2}}, c_{\ref{lem:smallhnorm}}\} \cdot k^4 \log n}{\phi_{in}^2}$, we know that the chosen parameters meet all the preconditions in these lemmas.

Since $G$ is $(k,\phi)$-clusterable, there exists some $h$, $1 \le h \le k$, and a partition of the vertex set of $G$ into $h$ subsets $C_1, \dots, C_h$, such that for every $i$, $1 \le i \le h$, we have $\phi(G[C_i]) \ge \phi$ and $\phi_G(C_i) \le c_{d,k} \varepsilon^4 \phi^2$. For any vertex $v$, define $C(v)$ to be the unique cluster $C_i$ to which $v$ belongs.

We call a vertex $v$ \textit{good} if the following three conditions are satisfied:
\begin{enumerate}
\item \label{item:norm} $\norm{\p_v^\ell}_2^2 \le \frac{48sk}{n}$.
\item \label{item:clustersize} $|C(v)| \ge \frac{1}{24ks}n$.
\item \label{item:expandcore} $v \in \widetilde{C(v)}$, where $\widetilde{C(v)} \subseteq C(v)$ is defined as in Lemma \ref{lem:smalll2} by setting $C = C(v)$.
\end{enumerate}

The success probability of the algorithm depends on the random coins of sampling and random walks. We show that with probability at least $\frac78$ (over random coins of sampling), all vertices in the sample set $S$ are good; and if all these vertices are good, then our tester will accept with probability at least $\frac56$ (over random coins of random walks). Together, this means that with probability at least $\frac78 \cdot \frac56
%= 1-\frac{13}{48}
= \frac{35}{48} \ge \frac23$ the tester will accept. This will conclude the proof of the lemma.

\begin{claim}
\label{claim:completeness-1}
With probability at least $\frac78$, all vertices in the sampled set $S$ are good.
\end{claim}

\begin{proof}
Let $v$ be any vertex that is sampled uniformly at random from $V$. By Lemma \ref{lem:smallhnorm}, the probability that $\norm{\p_v^\ell}_2^2 > \frac{48sk}{n}$ is at most $\alpha = \frac{1}{24s}$. Since there are at most $k$ clusters, the probability that $v$ belongs to a cluster of size at most $\frac{1}{24ks}n$ is at most the probability that $v$ is one of at most $k \cdot \frac{n}{24ks}$ vertices in these small clusters, which is $\frac{1}{24s}$. In addition, since $|\widetilde{C(v)}|\ge (1-\alpha)|C(v)|$, the probability that $v\notin \widetilde{C(v)}$ is at most $\alpha = \frac{1}{24s}$. Overall, the probability that $v$ is not good is at most $\frac{1}{24s} + \frac{1}{24s} + \frac{1}{24s} = \frac{1}{8s}$. By the above analysis and the union bound, with probability at least $1 - \frac{1}{8s} \cdot s = \frac78$, all sampled vertices in $S$ are good.
\end{proof}

\begin{claim}
\label{claim:completeness-2}
Conditioned on the event that all the sampled vertices $v\in S$ are good, our tester will accept $G$ with probability at least $\frac56$.
\end{claim}

\begin{proof}
Let $v\in S$. Since $v$ is good, then $\norm{\p_v^\ell}_2^2\le\frac{48sk}{n}=\frac{\sigma}{4}$. Now by Lemma \ref{lem:collision}, \textbf{$l_2^2$-norm estimator} will reject $v$ with probability at most $\frac{16\sqrt{n}}{r}\le\frac{1}{12s}$. By the union bound, the probability that we get rejected at step \ref{alg:l2norm} of the algorithm is at most $\frac{1}{12}$.

For any two vertices $u,v$ from $S$, if $u,v$ belong to the same large cluster, then by Conditions \ref{item:clustersize}--\ref{item:expandcore} of good vertices and by Lemma \ref{lem:smalll2}, $\norm{\p_u^\ell-\p_v^\ell}_2^2\le \frac{1}{4n}$. %, while if $u,v$ are from two different (large) clusters, then by Condition \ref{item:rwremain} and Lemma \ref{lem:largel2}, $\norm{\p_u^\ell-\p_v^\ell}_2^2\ge \frac{1}{n}$.
Now recall that we have set $b=\frac{216sk}{n}, \xi=\frac{1}{4n}, \delta=\frac{1}{12s^2}$ and $r = 192 c_{\ref{cor:distribution}} s \sqrt{skn} \ln s$ in Theorem \ref{cor:distribution}. Then $b\ge\max\{\norm{\p_v^\ell}_2^2,\norm{\p_u^\ell}_2^2\}$, $r \ge c_{\ref{cor:distribution}} \cdot \frac{\sqrt{b}}{\xi} \ln\frac{1}{\delta}$, and we can ensure that with probability at least $1-\delta$, any call to \textbf{$l_2$-Distribution-Test} will accept the distributions $\p_u^t,\p_v^t$ if $u,v$ belong to the same large cluster. %, and reject the distributions if $u,v$ are from two different (large) clusters.
By the union bound, the probability that there exist some call such that the distribution tester does not accept $u,v$ if $u,v$ are from the same cluster is at most $s^2\delta\le \frac{1}{12}$. Therefore, the probability that the algorithm does not reject at step \ref{alg:l2norm} and all the calls to the \textbf{$l_2$-Distribution-Test} return the correct answer is at least $1-\frac{1}{12}-\frac{1}{12}=\frac{5}{6}$.

Now note that if for any $u,v\in S$ such that $u,v$ belong to the same cluster, the distribution tester with input $\p_u^\ell,\p_v^\ell$ accepts, then there will an edge $(u,v)$ in the ``similarity graph'' $H$. %For vertices $u,v\in S$ that belong to different clusters, there will be no edge between them in $H$.
This further implies that all the vertices in $S$ that are in the same cluster will form a clique. %\Artur{Is this a clique?}.
(But note that two sampled vertices from two different clusters might also be connected in $H$.)
Since there are at most $k$ clusters, we will get at most $k$ connected components in $H$, and thus the tester will accepts $G$.
\end{proof}

%Finally, for any cluster $C$ with size at least $\varepsilon n$, the probability that no vertex from $C$ is sampled is at most $(1-\varepsilon)^s$. Furthermore, since there are at most $k$ clusters, the probability that there exists some cluster $C$ with $|C|\ge \varepsilon n$ and no vertex in $C$ is sampled is at most $k\cdot (1-\varepsilon)^s\le \frac{1}{16}$ by our choice of $s$. Therefore, with probability at least $1 - \frac{13}{48} - \frac{1}{16} = \frac23$, the algorithm will accept $G$ and for each cluster $C$ of size at least $\varepsilon n$, at least one vertex from $C$ will sampled.%\Artur{This will require polishing, as I said earlier, we need to describe clearly what are the representatives etc.}
%\Artur{As you may guess, I would prefer to present this paragraph differently, since I want us to have a different claim about the representatives, as I wrote in the statement of Theorem \ref{thm:main}.}
%Therefore, with probability at least $1- \frac{13}{48} \ge \frac 23$ the algorithm will accept $G$.

We can now apply Claims \ref{claim:completeness-1} and \ref{claim:completeness-2} to conclude the proof of Lemma \ref{lemma:completeness}.
\end{proof}

%=============================================================================

\subsubsection{Soundness --- rejecting graphs $\varepsilon$-far from $(k, \phi^*)$-clusterable}

We present now a proof of the soundness of our tester.

\begin{lemma}
\label{lemma:soundness}
Let $\gamma = \gamma_{d,k} > 0$ be some constant depending on $d, k$. If the input graph $G = (V,E)$ is $\varepsilon$-far from $(k, \phi^*)$-clusterable with $\phi^*\le \frac{\gamma \varepsilon^2}{s \ell}$, then the algorithm \textbf{$k$-Cluster-Test} rejects $G$ with probability at least $\frac23$.
\end{lemma}

\begin{proof}
We will use $\gamma = \min\{\frac{1}{48 c_{\ref{lemma:partition-eps-far-improved}}}, \alpha_{\ref{lemma:partition-eps-far-improved}}\}$. Let us first observe that our choice of $\gamma$ ensures that Lemma \ref{lemma:partition-eps-far-improved} implies the existence of a partition of $V$ into $k+1$ disjoint sets $V_1, \dots, V_{k+1}$ such that for each $i$, $1 \le i \le k+1$, $|V_i| \ge \kappa_1 \varepsilon^2 |V|$ and $\phi_G(V_i) \le \kappa_2 \phi^* \varepsilon^{-2}$, for appropriate parameters $\kappa_1 = \frac{1}{1152k}$ and $\kappa_2 = c_{\ref{lemma:partition-eps-far-improved}}$.

Let $\alpha = \frac{1}{24s}$ (here $\alpha$ corresponds to the parameter $\alpha$ used in Lemma \ref{lem:largel2}). For every set $V_i$, $1 \le i \le k+1$, let $\widehat{V_i} \subseteq V_i$ be the set of vertices $v \in V_i$ such that the probability that the random walk of length $\ell$ starting at $v$ does not leave $V_i$ is at least $1 - \frac{\kappa_2 \phi^* \ell}{2 \alpha \varepsilon^2}$. We observe that since $\phi_G(V_i) \le \kappa_2 \phi^* \varepsilon^{-2}$, we have $|\widehat{V_i}| \ge (1-\alpha) |V_i|$ by Lemma \ref{lem:largel2}. Hence, our assumption that $|V_i| \ge \kappa_1 \varepsilon^2 |V|$ implies that $|\widehat{V_i}| \ge (1-\alpha) \kappa_1 \varepsilon^2 |V|$.

Let us call the sample set $S$ chosen by the algorithm \textbf{$k$-Cluster-Test} to be \emph{representative} if $\widehat{V_i} \cap S \ne \emptyset$ for every $i$, $1 \le i \le k+1$, and $S\subseteq \bigcup_{i=1}^{k+1} \widehat{V_i}$.

\begin{claim}
\label{claim:soundness-1}
The probability that the sample set $S$ is representative is at least $\frac56$.
\end{claim}

\begin{proof}
For any set $X \subseteq V$, $\Pr[X \cap S = \emptyset] = (1-|X|/|V|)^{s} \le e^{- s |X|/|V|}$. Therefore, since $|\widehat{V_i}| \ge (1-\alpha) \kappa_1 \varepsilon^2 |V|$, the probability that $S$ does not contain any element from $\widehat{V_i}$ is smaller than or equal to $e^{- s \widehat{V_i}/|V|} \le e^{- s (1-\alpha) \kappa_1 \varepsilon^2}$. Hence, the union bound implies that the probability that there exists some $i\le k+1$ such that $S$ does not contain any element from $\widehat{V_i}$ is at most $(k+1) \cdot e^{- s (1-\alpha)\kappa_1 \varepsilon^2}$. In addition, the probability that there exists some vertex in $S$ that belongs to $V\setminus(\bigcup_{i=1}^{k+1} \widehat{V_i})$ is at most $s\cdot\alpha$. Therefore, the probability that $S$ is representative is greater than or equal to $1 - (k+1) \cdot e^{- s (1-\alpha) \kappa_1 \varepsilon^2} - s \alpha$. Since $s = \frac{1536k \ln(8(k+1))}{\varepsilon^2}$ and $\alpha = \frac{1}{24s}$, we have $s (1-\alpha) \kappa_1 \varepsilon^2 \ge \ln (8(k+1))$, and hence we can conclude that this probability is at least $\frac56$.
\end{proof}

\begin{claim}
\label{claim:soundness-2}
If $S$ is representative then the algorithm \textbf{$k$-Cluster-Test} rejects $G$ with probability at least $\frac56$.
\end{claim}

\begin{proof}
Let $S_i:=\widehat{V_i}\cap S$. Since $S$ is representative, then $S = \bigcup_{i=1}^{k+1} S_i$. Recall that the algorithm \textbf{$k$-Cluster-Test} rejects $G$ if one of the following two cases happen:
\begin{itemize}
\item there is a $v \in S$ such that \textbf{$l_2^2$-norm estimator} passes the testing of $\norm{\p_v^t}_2^2 > \sigma$. %\Pan{this should be modified by the new description of the distribution tester}, or
\item for any $1\le i< j\le k+1$, and any vertex pair $u,v$ such that $u\in S_i$ and $v\in S_j$, $(u,v)$ %there exist $v_1, \dots, v_{k+1} \in S$ such that for every $1 \le i < j \le k+1$, $(v_i,v_j)$
    is not an edge in the ``similarity graph'' (because in that case the resulting graph $H$ could not be a union of at most $k$ connected components).
\end{itemize}

%Since $S$ is representative then, for every $i$, $1 \le i \le k+1$, let us choose as $v_i$ any element from $\widehat{V_i} \cap S$.
If there exists some $v \in S$ with $\norm{\p_v^t}_2^2 > \sigma$, then by Lemma \ref{lem:collision}, \textbf{$l_2^2$-norm tester} with rejects $v$ with probability at least $1-\frac{16\sqrt{n}}{r}>\frac{2}{3}$ and we are done. Therefore, we assume in the following that for every $v\in S$, $\norm{\p_v^t}_2^2 < \sigma$. Let us now observe that the probability that the algorithm \textbf{$k$-Cluster-Test} would reject $G$ is lower bounded by the probability that for any $1 \le i < j \le k+1$, and any vertex pair $u,v$ such that $u\in S_i$ and $v\in S_j$, \textbf{$l_2$-Distribution-Test} rejects the distributions $\p_{u}^\ell,\p_{v}^\ell$.
%(Note that here in the $l_2$-norm estimator, we do not require that $Z_v \le \zeta$ for every $v \in S$. Indeed, if there is a $v \in S$ with $Z_v > \zeta$ then the algorithm \textbf{$k$-Cluster-Test} would reject $G$ anyway, and so, the output of \textbf{$l_2$-Distribution-Test}($r, F_{v_i}, F_{v_j}$) is irrelevant. Therefore, the probability that for every pair of vertices $v_i, v_j$, $1 \le i < j \le k+1$, \textbf{$l_2$-Distribution-Test}($r, F_{v_i}, F_{v_j}$) outputs \textbf{Rej} is smaller than or equal to the probability that either the tester would reject $G$ is Step 3 or if $Z_v \le \zeta$ for every $v \in S$, then it would reject in Step 5.)

Our definition of sets $\widehat{V_1}, \widehat{V_2}, \dots, \widehat{V_{k+1}}$ and the assumption on $\phi_{in}^*$ (which implies that $\ell \le \frac{\alpha}{2 \kappa_2 \phi^* \varepsilon^{-2}} \le \frac{\alpha}{2 \max_i\{\phi_G(V_i)\}}$) ensure that for any $1 \le i < j \le k+1$, and any vertex pair $u,v$ such that $u\in S_i$ and $v\in S_j$, we can apply Lemma \ref{lem:largel2} to obtain $\norm{\p_{u}^{\ell}-\p_{v}^{\ell}}_2^2 \ge \frac{1}{n}$. We know, by Theorem \ref{cor:distribution} and our choice of $b, \xi, \delta$ in that theorem, that for every such pair $v_i$, $v_j$, \textbf{$l_2$-Distribution-Test} will accept the distributions $\p_{v_i}^\ell,\p_{v_j}^\ell$ with probability at most $\delta$. Therefore, the probability that there exists some vertex pair $u,v$ such that $u\in S_i$, $v\in S_j$, $1\le i<j\le k+1$ and $(u,v)$ is selected as an edge in the ``similarity graph'' (which would mean that \textbf{$l_2$-Distribution-Test} will accept $\p_{u}^\ell,\p_{v}^\ell$) is at most $s^2 \cdot \delta$. Therefore we can conclude that the algorithm \textbf{$k$-Cluster-Test} rejects $G$ with probability at least $1 - s^2 \cdot \delta \ge \frac56$.
\end{proof}

Now, the proof of Lemma \ref{lemma:soundness} follows directly from Claims \ref{claim:soundness-1} and \ref{claim:soundness-2}.
\end{proof}

We set $c' \frac{\phi^2 \varepsilon^4}{\log n} \le \frac{\gamma \varepsilon^2}{s \ell}$ in Theorem \ref{thm:main}. By our choice of $s$ and $\ell$, we can find a constant $c' = c'_{d,k}$ that depends on $d$ and $k$ satisfying this condition, and we then require that $\phi^* \le c' \frac{\phi^2 \varepsilon^4}{\log n}$.

%\begin{remark}
%Note that by Lemma~\ref{lemma:partition-eps-far-improved} and the above proof, our algorithm actually rejects any graph that is $\varepsilon$-far from $(k, \phi^*_{in}, \phi^*_{out})$-clusterable for $\phi^*_{in} \le c' \frac{\phi_{in}^2 \varepsilon^4}{\log n}$ and any $\phi^*_{out} \ge 0$.
%\end{remark}

%=============================================================================

\subsubsection{Running time}

Now we analyze the running time of the algorithm \textbf{$k$-Cluster-Test}. First note that to sample from distributions $\p_v^\ell$ for any $v\in V$, we need to perform $r$ random walks of length $\ell$ from $v$ and the corresponding time is $O(\ell r)$. Note that each invocation of either distribution tester runs in time linearly in the number of samples, that is $r$. Since we sampled $s$ vertices, invoked \textbf{$l_2^2$-norm tester} for each vertex in the sample set $S$, and invoked \textbf{$l_2$-Distribution-Test} for each vertex pair in $S$, we know that the total running time of the algorithm is $O(\ell s r + r s + s^2 \frac{\sqrt{b}}{\xi} \ln\frac{1}{\delta}) = O(\frac{\sqrt{n} k^7 (\ln k)^{7/2} \ln\frac{1}{\varepsilon}\ln n}{\phi_{in}^2 \varepsilon^5})$.

This completes the proof of Theorem \ref{thm:main}, which follows directly from Lemmas \ref{lemma:completeness} and \ref{lemma:soundness}, and our analysis of the running time given above.

%=============================================================================

\section{Proofs of central properties (Lemmas \ref{lem:smalll2} -- \ref{lemma:partition-eps-far-improved})}
\label{sec:proof-section5}

In the following, we will prove Lemmas
%\ref{lem:smalll2}, \ref{lem:smallhnorm}, \ref{lem:largel2}, \ref{lemma:partition-eps-far-improved}.
\ref{lem:smalll2} -- \ref{lemma:partition-eps-far-improved}.
Before that, we present two spectral property on the eigenvalues of $(k, \phi_{in}, \phi_{out})$-clusterable graphs, which might be of independent interest.

%=============================================================================

\subsection{Spectral properties of clusterable graphs}
\label{subsec:spectral-clusterable}

Before we state the spectral properties of clusterable graphs, we first observe that it will be sufficient for us to consider weighted $d$-regular clusterable graphs. This is true since our algorithm actually performs the lazy random walk on the (virtual) weighted $d$-regularized version $G_\reg$ of the input $d$-bounded degree graph $G$. In addition, under our definition, for any set $S\subseteq V$, the outer conductance $\phi_G(S)$ and inner conductance $\phi(G[S])$ of $S$ in $G$ are the same as outer conductance $\phi_{G_\reg}(S)$ and inner conductance $\phi(G_\reg[S])$ of $S$ in $G_\reg$, respectively. For this reason, in the rest of this section, we will assume that $G$ is a weighted $d$-regular graph.

The proofs of spectral properties of clusterable graphs rely on a recent high-order Cheeger inequality by Lee et al.\ \cite{LOT12:high}. To state the inequality, we first introduce some notations.

Let $\A$ denote the adjacency matrix of $G$. Let $\LL = \I - \frac{1}{d} \A$ be the Laplacian matrix of $G$, where $\I$ is the identity matrix. Let $\lambda_i$ be the $i$th smallest eigenvalue of the Laplacian matrix $\LL$ and let $\vv_i$ denote the corresponding (unit) eigenvector. Note that the probability transition matrix of the lazy random walk on $G$ is $\W := \frac{\I + \frac{1}{d} \A}{2}$, and it is straightforward to see that $\{1-\frac{\lambda_i}{2}\}_{1 \le i \le n}$ is the set of eigenvalues of $\W$ with corresponding eigenvectors  $\{\vv_i\}_{1 \le i \le n}$ (cf. Appendix \ref{subsec:spectra} for more details).

For a $d$-regular graph $G$, let $\rho_G(k)$ denote the minimum value of the maximum conductance over any possible $k$ disjoint nonempty subsets. That is,
\begin{displaymath}
    \rho_G(k)
        :=
    \min_{\textrm{disjoint $S_1,\dots, S_k$}}\max_{1\le i\le k}\phi_G(S_i)
        \enspace.
\end{displaymath}

Lee et al.\ \cite{LOT12:high} proved the following higher-order Cheeger's inequality.

\begin{theorem}[\cite{LOT12:high}]
\label{thm:highcheeger}
For any weighted $d$-regular graph $G$ and any $k \ge 2$, it holds that
\begin{displaymath}
    \lambda_k/2
        \le
    \rho_G(k)
        \le
    c_{\ref{thm:highcheeger}} k^2 \sqrt{\lambda_k}
        \enspace,
\end{displaymath}
where $c_{\ref{thm:highcheeger}}$ is some universal constant.
\end{theorem}

\begin{remark}
Lee et al.\ actually proved a stronger version of the above theorem that applies to any weighted graph, by using a \emph{volume}-based definition of conductance (see Appendix \ref{subsec:volume-definition}). The weaker version given by Theorem \ref{thm:highcheeger} will be enough for our application.
%We also remark that our proof of the following spectral properties can be generalized to arbitrary weighted graphs by using the volume-based definition of conductance and the stronger version of the above theorem (i.e. Theorem~\ref{thm:highcheeger-weighted}). We omit the details here.
\end{remark}

Now we are ready to state the spectral properties of clusterable graphs, which are given in the following two lemmas. The first lemma says that in a $k$-clusterable graph there is a large gap between $\lambda_h$ and $\lambda_{h+1}$ for some $h \le k$.

\begin{lemma}
\label{lem:eigenvalue}%\Pan{The notion of exactly $h$-clusterable is for this lemma. Maybe we can remove this ``exactly clusterable'' definition.}
%I tried to generalize our tester to general graphs, but this lemma seems do not generalize, in particular, $\phi_G(P)=\frac{e(P,V\setminus P)}{\vol(P)}$ may be very small as the denominator (involving the edges between different $B_i$'s) may be large.}
If $G$ is weighted $d$-regular and $(k, \phi_{in}, \phi_{out})$-clusterable, then there exists $h$, $1 \le h \le k$, such that $\lambda_i\le 2\phi_{out}$ for any $i\le h$, and $\lambda_{i} \ge \frac{\phi_{in}^2}{c_{\ref{thm:highcheeger}}^2h^4}$ for any $i \ge h+1$.
\end{lemma}

\begin{proof}%[Proof of Lemma \ref{lem:eigenvalue}]
Since $G$ is $(k,\phi_{in},\phi_{out})$-clusterable, then for some $h$, $1 \le h \le k$, there exists a partition of $V$ into $h$ sets $C_1, \dots, C_h$, such that $\phi(G[C_i]) \ge \phi_{in}$ and $\phi_G(C_i) \le \phi_{out}$ for any $i \le h$. From the latter, we obtain that $\rho_G(h) \le \max_i \phi_G(C_i) \le \phi_{out}$ and then by Theorem \ref{thm:highcheeger}, $\lambda_h \le 2 \phi_{out}$, and thus for any $i \le h$, $\lambda_i \le \lambda_h\le 2\phi_{out}$.

Next, let us consider an arbitrary $(h+1)$-partition $P_1, \dots, P_{h+1}$ of $V$. We note that there must be at least one set in the partition, say $P_{i_0}$, such that $|P_{i_0} \cap C_j| \le \frac12 |C_j|$ for every $1\le j\le h$. This is true since otherwise, for every $i$, $1 \le i \le h+1$, each $P_i$ would contain more than half of the vertices of some cluster, say $C_{\pi(i)}$, that is, $|P_i \cap C_{\pi(i)}| > \frac12 |C_{\pi(i)}|$. Then, since there are %at most
$h$ clusters $C_1, \dots, C_h$, by the pigeonhole principle there would have to exist two indices $i$ and $j$, $1 \le i < j \le h+1$, such that $\pi(i) = \pi(j)$. This would mean that each of $P_i$ and $P_j$ contain more than half of the vertices from the same cluster $C_{\pi(i)}$, which is a contradiction since $P_i$ and $P_j$ are disjoint. This proves the existence of the set $P_{i_0}$.

Let $P := P_{i_0}$. For every $1 \le i \le h$, let $B_i := P \cap C_i$. Since each cluster $C_i$ has large inner conductance, namely $\phi(G[C_i]) \ge \phi_{in}$, and since $|B_i| \le \frac12 |C_i|$, we have $e(B_i, C_i \setminus B_i) \ge \phi_{in} d |B_i|$ for every $1 \le i \le h$. Hence, $\phi_G(P) = \frac{e(P, V \setminus P)}{d |P|} \ge \frac{\sum_{i=1}^h e(B_i, C_i \setminus B_i)}{d \sum_{i=1}^h |B_i|} \ge \phi_{in}$, and thus $\rho_G(h+1) \ge \phi_{in}$. Therefore Theorem \ref{thm:highcheeger} gives $\phi_{in} \le \rho_G(h+1) \le c_{\ref{thm:highcheeger}} h^2 \sqrt{\lambda_{h+1}}$, which yields $\lambda_{h+1} \ge \frac{\phi_{in}^2}{c_{\ref{thm:highcheeger}}^2 h^4}$.
\end{proof}

The second lemma states that in a $k$-clusterable graph, for any large cluster $C$, the average value of $(\vv_i(u)-\vv_i(v))^2$ over all $|C|^2$ vertex pairs $u,v\in C$ is as small as $\Theta_d(\frac{\phi_{out}}{|C|\phi_{in}^2})$,
%for any $i\le h\le k$, where $h$ is guaranteed as in the above lemma.
for any $i\le h\le k$.

\begin{lemma}
\label{lem:clusterable-eigenvector}
Let $G = (V,E)$ be a weighted $d$-regular graph that is $(k, \phi_{in}, \phi_{out})$-clusterable and let $C \subseteq V$ be any subset with $\phi(G[C])\ge\phi_{in}$. Then there is $h$, $1 \le h \le k$ such that for every $i$, $1 \le i \le h$, the following holds:
\begin{displaymath}
    \frac{1}{|C|} \sum_{u, v \in C} (\vv_i(u) - \vv_i(v))^2
        \le
    \frac{8 d^4 \phi_{out}}{\phi_{in}^2}
    \enspace.
\end{displaymath}
\end{lemma}

\begin{proof}%[Proof of Lemma \ref{lem:clusterable-eigenvector}]
Since $G$ is $(k, \phi_{in}, \phi_{out})$-clusterable, by Lemma \ref{lem:eigenvalue}, there exists $h$, $1 \le h \le k$, such that $\lambda_{h+1} \ge \frac{\phi_{in}^2}{c_{\ref{thm:highcheeger}}^2 h^4}$ and $\lambda_i \le 2 \phi_{out}$ for any $1 \le i \le h$. Hence, for any $i \le h$, by the variational principle of eigenvalues (see Fact \ref{fact:eigenvalue} in Appendix), we have
\begin{eqnarray}
\label{eqn:lambdaphiout}
    \lambda_i
        =
    \frac{\sum_{(u,v) \in E} (\vv_i(u) - \vv_i(v))^2}{d}
        \le
    2 \phi_{out}
    \enspace.
\end{eqnarray}

Let us recall a known result (see, e.g., \cite[(1.5), p.~5]{Chu97:spectral}) that for any weighted graph $H = (V_H, E_H)$,\footnote{We remark that in \cite{Chu97:spectral}, the summation in the denominator is over all unordered pairs of vertices, while in our context, the summation is over all possible $|V_H|^2$ vertex pairs. Therefore, a multiplicative factor $2$ appears in the numerator in equation (\ref{ineq:from-Chung}), compared with the form in \cite[(1.5), p.~5]{Chu97:spectral}.} %\Pan{(Here the definition of $d_H(v)$ is not consistent with the previous one $d_G(v)$)}
\begin{equation}
\label{ineq:from-Chung}
    \lambda_2(H)
        =
    \vol_H(V_H) \cdot
    \min_{f} \left\{\frac{2\cdot \sum_{(u,v) \in E_H} (f(u) - f(v))^2}
        {\sum_{u,v \in V_H} (f(u) - f(v))^2 d_H(u) d_H(v)}
        \right\}
    \enspace,
\end{equation}
where $\lambda_2(H)$ denotes the second smallest eigenvalue of the normalized Laplacian of $H$,
%the summation in the dominator is over all unordered pairs of $u,v \in V_H$, and
the volume $\vol_H(S)$ of a set $S\subseteq V_H$ is the sum of degrees of vertices in $S$, that is, $\vol_H(S) := \sum_{v \in S} d_H (v)$.

Let us consider the induced subgraph $H := G[C]$ on $C$. Let $\phi_H^\vol(S) := \frac{e(S,H\setminus S)}{\vol_H(S)}$ and $\phi^\vol(H) := \min_{S: \vol_H(S) \le \vol_H(V_H)/2}\frac{e(S,H\setminus S)}{\vol_H(S)}$ (cf. Appendix \ref{subsec:volume-definition}). Since $\phi(H) \ge \phi_{in}$, then it is straightforward to see that $\phi^\vol(H) \ge \frac{\phi_{in}}{d}$\footnote{This can be verified by considering the set $S$ with $\vol_H(S) \le \vol_H(V_H)/2$ such that $\phi_H^\vol(S) = \phi^\vol(H)$: if $|S| \le \frac{|V_H|}{2}$, then $\phi_H^\vol(S)\ge \phi_H(S)\ge \phi_{in}$; if $|S| > \frac{|V_H|}{2}$, then $\phi_H^\vol(S)\ge \frac{e(S,V_H\setminus S)}{d|S|}\ge \frac{\phi_{in}d |V_H\setminus S|}{d|S|}\ge \frac{\phi_{in}}{d}$, where the penultimate inequality follows from the fact that $\phi_H(V_H\setminus S)=\frac{e(S,V_H\setminus S)}{d|V_H\setminus S|}\ge \phi_{in}$ and the last inequality follows from that $|S|\le \vol_H(S)\le \vol_H(V_H\setminus S)\le d|V_H\setminus S|$.}. Cheeger's inequality (cf. Theorem \ref{thm:cheeger}) yields $\lambda_2 (H) \ge \frac{\phi_{in}^2}{2d^2}$. Therefore, if we apply this bound to inequality (\ref{ineq:from-Chung}), then,
\begin{displaymath}
    \vol_H(V_H) \cdot \frac{2\cdot\sum_{(u,v) \in E_H}(\vv_i(u) - \vv_i(v))^2}
        {\sum_{u,v \in V_H} (\vv_i(u) - \vv_i(v))^2 d_H(u) d_H(v)}
        \ge
    \lambda_2(H)
        \ge
    \frac{\phi_{in}^2}{2d^2}
    \enspace.
\end{displaymath}

Combining this with the fact that $\sum_{(u,v) \in E_H}(\vv_i(u) - \vv_i(v))^2 \le \sum_{(u,v) \in E_G} (\vv_i(u) - \vv_i(v))^2 \le 2 d \phi_{out}$, where the last inequality follows from inequality (\ref{eqn:lambdaphiout}), we have that
\begin{displaymath}
    \sum_{u,v \in V_H} (\vv_i(u) - \vv_i(v))^2 d_H(u) d_H(v)
        \le
    \frac{8 d^3 \vol_H(V_H) \phi_{out}}{\phi_{in}^2}
    \enspace.
\end{displaymath}
Next, since $\phi(H) \ge \phi_{in} > 0$ implies that $d_H(u) \ge 1$ for any $u \in V_H$, and since the fact that for any $u \in V_H$, $d_H(u) \le d$ yields $\vol_H(V_H) \le d|V_H| = d|C|$, using the bound above we obtain:
\begin{displaymath}
    \sum_{u,v \in V_H} (\vv_i(u) - \vv_i(v))^2
        \le
    \sum_{u,v \in V_H} (\vv_i(u) - \vv_i(v))^2 d_H(u) d_H(v)
        \le
    \frac{8 d^3 \vol_H(V_H) \phi_{out}}{\phi_{in}^2}
        \le
    \frac{8 d^4 |C| \phi_{out}}{\phi_{in}^2}
    \enspace.
\end{displaymath}
The completes the proof of Lemma \ref{lem:clusterable-eigenvector}.
\end{proof}

\begin{remark}
In Lemma \ref{lem:counterexample} we show that Lemma \ref{lem:clusterable-eigenvector} is essentially tight for $k=2$ and constant $\phi_{in}$. We prove that there is a $(2,\phi_{in},\phi_{out})$-clusterable graph $G$ with clusters $C_1,C_2$ such that for at least one cluster, say $C_1$, the average value of $(\vv_2(u)-\vv_2(u))^2$ between vertices $u,v$ from $C_1$ is $\Omega(\frac{\phi_{out}}{d^3|C_1|})$.
\end{remark}

%\begin{remark}
%Note that we can relax the conditions in the completeness to be any $(k,\phi_{in},\phi_{out})$-graph such that $\lambda_{k}\le \frac$
%\end{remark}

%=============================================================================

\subsection{Proofs of Lemmas \ref{lem:smalll2}, \ref{lem:smallhnorm}, \ref{lem:largel2}}
\label{sec:proofs}

In this section, we prove Lemmas \ref{lem:smalll2} -- \ref{lem:largel2}. For a $d$-bounded degree graph $G$, recall that $\p_v^t$ is the probability distribution of the endpoints of the lazy random walk of length $t$ starting from $v$ on $G_\reg$. Let $\W_\reg$ be the probability transition matrix of the lazy random walk on $G_\reg$ and let $\1_v$ be the characteristic vector on vertex $v$. Then $\p_v^t = \1_v(\W_\reg)^t$.

In this section, let $\lambda_i^{\reg}$ denote the $i$th smallest eigenvalue of the normalized Laplacian matrix of the regularized version $G_\reg$ of $G$ and let $\vv_i^{\reg}$ be the corresponding unit eigenvector.

%=============================================================================

Now we prove Lemma \ref{lem:smalll2}, which shows that the $l_2$-norm of the difference of two random walk distributions $\p_v^t - \p_u^t$ is small for most pairs $u, v$ from the same cluster for $t$ large enough.

\begin{proof}[Proof of Lemma \ref{lem:smalll2}]
For the $d$-bounded degree graph $G$, we apply Lemma \ref{lem:clusterable-eigenvector} to its weighted $d$-regular version $G_\reg$.
% with $\phi_{in} = \phi$ and $\phi_{out} \le c_{d,k} \varepsilon^4 \phi^2/\log n$, where $c_{d,k} = \frac{c}{d^2 k^5 \ln^2 k}$ for a sufficiently small universal constant $c > 0$.
For the subset $C$, by defining $\centr_{C,i} := \frac{1}{|C|} \sum_{u \in C} \vv_i^{\reg}(u)$, we obtain the following:
%\Artur{Since we're using $c_{?,??}$ to denote constants that either depend on ? and ??? or that are defined in Claim/Lemma ?.??, I changed the notation and have now $c_{C,i} \rightarrowtail \centr_{C,i}$. And please feel free to change it to your own taste (it's defined somewhere at the beginning of the file); I took $\centr$ somehow at random, as something similar to $c$ but not too close.}
%
\begin{displaymath}
    \sum_{u \in C} (\vv_i^{\reg}(u) - \centr_{C,i})^2
        =
    \frac{1}{|C|} \sum_{u,v \in C} (\vv_i^{\reg}(u) - \vv_i^{\reg}(v))^2
        \le
    \frac{4 d^4 \phi_{out}}{\phi_{in}^2}
        %\le
    %c_{d,k} \varepsilon^4
        \enspace,
\end{displaymath}
where we used the elementary identity $\frac{1}{n} \sum_{i<j} (a_i - a_j)^2 = \sum_{i=1}^n (a_i - \frac{\sum_{i=1}^n a_i}{n})^2$ for any $a_1, \dots, a_n$.

Therefore, the average of $(\vv_i^{\reg}(u) - \centr_{C,i})^2$ over all vertices in $C$ is at most $\frac{1}{|C|} \cdot \frac{4 d^4 \phi_{out}}{\phi_{in}^2}
%= \frac{8 d^2 \phi_{out}}{(|C|-1) \phi_{in}^2}\le \frac{4 d^2 \phi_{out}}{|C| \phi_{in}^2}
$. This implies that for at least $(1-\alpha) |C|$ vertices $u \in C$, we have $(\vv_i^{\reg}(u) - \centr_{C,i})^2 \le \frac{4 k d^4 \phi_{out}}{\alpha |C| \phi_{in}^2}$ for all $i$, $1 \le i \le h \le k$. Let $\widetilde{C} \subseteq C$ denote the set of vertices with this property.

Consider any two vertices $u, v \in \widetilde{C}$.
%
%\Artur{I see that this ($\norm{x-y}_2^2 \le 2(\norm{x-z}_2^2 + \norm{z-y}_2^2)$) is true, but is this is a well-known fact? If not, then one should mention something about it.}
%
We observe that for any $i$, $1 \le i \le h$, we have $(\vv_i^{\reg}(u) - \vv_i^{\reg}(v))^2 \le 2 ((\vv_i^{\reg}(u) - \centr_{C,i})^2 + (\vv_i^{\reg}(v) - \centr_{C,i})^2) \le \frac{16 kd^4 \phi_{out}}{\alpha |C| \phi_{in}^2}$, where the first inequality that $(x-y)^2\le 2((x-z)^2+(z-y)^2)$ follows directly from the Cauchy-Schwarz inequality, and the second inequality follows from the property of vertices in $\widetilde{C}$. Next, by Fact \ref{fact:spectra} we have $\p_v^t - \p_u^t = \sum_{i=1}^n (\vv_i^{\reg}(v) - \vv_i^{\reg}(u))(1 - \frac{\lambda_i^{\reg}}{2})^t \vv_i^{\reg}$, and therefore
\begin{eqnarray*}
\norm{\p_v^t - \p_u^t}_2^2
    & = &
\sum_{i=1}^n(\vv_i^{\reg}(u)-\vv_i^{\reg}(v))^2(1-\frac{\lambda_i^{\reg}}{2})^{2t}\\
    &=&
\sum_{i=1}^h(\vv_i^{\reg}(u)-\vv_i^{\reg}(v))^2(1-\frac{\lambda_i^{\reg}}{2})^{2t}+
    \sum_{i=h+1}^n(\vv_i^{\reg}(u)-\vv_i^{\reg}(v))^2(1-\frac{\lambda_i^{\reg}}{2})^{2t}\\
    &\le&
\sum_{i=1}^h(\vv_i^{\reg}(u)-\vv_i^{\reg}(v))^2 +
    (1-\frac{\lambda_{h+1}^{\reg}}{2})^{2t}\sum_{i=h+1}^n(2\vv_i^{\reg}(u)^2+2\vv_i^{\reg}(v)^2)\\
    & \le &
\frac{16 h kd^4 \phi_{out}}{\alpha |C| \phi_{in}^2} + 4 (1-\frac{\phi_{in}^2}
        {2c_{\ref{thm:highcheeger}}^2h^4})^{2t}\\
    & \le &
\frac{16 k^2 d^4 \phi_{out}}{\alpha \beta n \phi_{in}^2} +
    4 (1 - \frac{\phi_{in}^2}{2 c_{\ref{thm:highcheeger}}^2 k^4})^{2t}
    \enspace.
\end{eqnarray*}
In the bound above, in the penultimate inequality we use the fact that
%$(\vv_i(u) - \vv_i(v))^2 \le \frac{8 k d^2 \phi_{out}}{\alpha |C| \phi_{in}^2}$ for all $1 \le i \le h$ and that
$\sum_{i=h+1}^n \vv_i^{\reg}(u)^2 \le \sum_{i=1}^n \vv_i^{\reg}(u)^2 = 1$ for any $u \in V$ (by Fact \ref{fact:spectra}) and $\lambda_{h+1}^{\reg} \ge \frac{\phi_{in}^2}{c_{\ref{thm:highcheeger}}^2 h^4}$ (by Lemma \ref{lem:eigenvalue}), and in the last inequality we use that $|C| \ge \beta n$.
%
%Now, we observe that for $\phi_{out} \le \frac{\alpha \beta \phi_{in}^2}{128 k d^2}$ we will have $\frac{16 k d^2 \phi_{out}}{\alpha \beta n \phi_{in}^2} \le \frac1{8n}$, and for $t \ge \frac{c_{\ref{thm:highcheeger}}^2 k^4 \log (32n)} {\phi_{in}^2}$ we will have $4 (1 - \frac{\phi_{in}^2}{2 c_{\ref{thm:highcheeger}}^2 k^4})^{2t} \le \frac1{8n}$.
%
Now by defining $\alpha_{\ref{lem:smalll2}}:=\alpha_{\ref{lem:smalll2}}(\alpha,\beta,d,k) =\frac{\alpha\beta}{128k^2d^4}$, $c_{\ref{lem:smalll2}}:=c_{\ref{thm:highcheeger}}^2$ and letting $t \ge \frac{c_{\ref{lem:smalll2}} k^4 \log n}{\phi_{in}^2}$, we can conclude that $\norm{\p_v^t-\p_u^t}_2^2 \le \frac{1}{4n}$.
%\Artur{Pan, since you're using these crazy constants, why do we need $\alpha \beta \ge 32 k^2 d^2 c_{d,k} \varepsilon^4$? Since we also have a $\log n$ in the denominator, we could do with a stronger value for $\alpha \beta$ \dots Anyway, IMO these games with constants make the paper overcomplicated and messy.}
\end{proof}

%=============================================================================

To prove Lemma \ref{lem:smallhnorm}, we again use the eigen-decomposition of vector $\p_u^t$ as given in Fact \ref{fact:spectra} and the fact that all eigenvalues of the normalized Laplacian of $G_\reg$ are large except for the first few ones. This allows us to bound the $l_2^2$ norm of $\p_u^t$ by its projection on the first few eigenvectors.

\begin{proof}[Proof of Lemma \ref{lem:smallhnorm}]
For any vertex $u \in V$, let $\delta(u) := \sum_{i=1}^k \vv_i^{\reg}(u)^2$. Since each eigenvector $\vv_i^{\reg}$ is of unit length, we have
\begin{eqnarray*}
    \sum_{u \in V} \delta(u)
        =
    \sum_{u \in V} \sum_{i}^k \vv_i^{\reg}(u)^2
        =
    \sum_{i}^k \sum_{u\in V} \vv_i^{\reg}(u)^2
        =
    k
    \enspace.
\end{eqnarray*}

Therefore, the expected value of $\delta(u)$ is at most $\frac{k}{n}$, and by the Markov's inequality, we know that for any $0 < \alpha < 1$, there exists a subset $V' \subseteq V$ such that $|V'| \ge (1-\alpha)|V|$ and that for any $u \in V'$, $\delta(u) \le \frac{k}{\alpha n}$. In addition, by Fact \ref{fact:spectra}, $\1_u = \sum_{i=1}^n \vv_i^{\reg}(u) \vv_i^{\reg}$, and $\p_u^t = \sum_{i=1}^n \vv_i^{\reg}(u) (1 - \frac{\lambda_i^{\reg}}{2})^t \vv_i^{\reg}$. Therefore,
\begin{eqnarray*}
    \norm{\p_u^t}_2^2
        =
    \norm{\sum_{i=1}^n \vv_i^{\reg}(u) (1-\frac{\lambda_i^{\reg}}{2})^t \vv_i^{\reg}}_2^2
        &=&
    \sum_{i=1}^n\vv_i^{\reg}(u)^2(1 -\frac{\lambda_{i}^{\reg}}{2})^{2t}
        \\
        &=&
    \sum_{i=1}^k\vv_i^{\reg}(u)^2(1- \frac{\lambda_i^{\reg}}{2})^{2t}+ \sum_{i=k+1}^n\vv_i^{\reg}(u)^2(1- \frac{\lambda_i^{\reg}}{2})^{2t}
        \\
        & \le &
    \sum_{i=1}^k\vv_i^{\reg}(u)^2+ (1-\frac{\lambda_{k+1}^{\reg}}{2})^{2t} \sum_{i=k+1}^n\vv_i^{\reg}(u)^2
        \\
        & \le &
    \delta(u)+(1-\frac{\lambda_{k+1}^{\reg}}{2})^{2t}
        \\
        & \le &
    \frac{k}{\alpha n} + (1 - \frac{\phi_{in}^2}{2 c_{\ref{thm:highcheeger}}^2 k^4})^{2t}
    \enspace,
\end{eqnarray*}
where in the last inequality, we used the fact that $\lambda_{k+1}^{\reg} \ge \frac{\phi_{in}^2}{c_{\ref{thm:highcheeger}}^2 k^4}$ by Lemma \ref{lem:eigenvalue}. In particular, the last bound implies that if $t \ge \frac{c_{\ref{lem:smallhnorm}} k^4 \log n}{\phi_{in}^2}$ for $c_{\ref{lem:smallhnorm}}:=c_{\ref{thm:highcheeger}}^2$, then $\norm{\p_u^t}_2^2 \le \frac{2 k}{\alpha n}$.%\Artur{I don't see why we're ending up with constant 16. What about: if $t \ge \frac{c_{\ref{thm:highcheeger}}^2 k^4 \ln(\alpha n/k)}{\phi_{in}^2}$, then $e^{\frac{- t \phi_{in}^2}{c_{\ref{thm:highcheeger}}^2 k^4}} \le \frac{k}{\alpha n}$. Hence, if $t \ge \frac{c_{\ref{thm:highcheeger}}^2 k^4 \log n}{\phi_{in}^2}$, then also $t \ge \frac{c_{\ref{thm:highcheeger}}^2 k^4 \ln(\alpha n/k)}{\phi_{in}^2}$, and therefore, by the bounds above:
%
%\begin{eqnarray*}
%    \norm{\p_u^t}_2^2
%        & \le &
%    \frac{k}{\alpha n} +
%        (1 - \frac{\phi_{in}^2}{2 c_{\ref{thm:highcheeger}}^2 k^4})^{2t}
%        \, \le \,
%    \frac{k}{\alpha n} +
%        e^{\frac{- 2t \phi_{in}^2}{2 c_{\ref{thm:highcheeger}}^2 k^4}}
%        \, = \,
%    \frac{k}{\alpha n} +
%        e^{\frac{- t \phi_{in}^2}{c_{\ref{thm:highcheeger}}^2 k^4}}
%        \, \le \,
%    \frac{k}{\alpha n} + \frac{k}{\alpha n}
%    \enspace.
%\end{eqnarray*}
%}
\end{proof}
%\subsection{Soundness}\label{subsec:soundness}

%=============================================================================

Now we give the proof of Lemma \ref{lem:largel2}, which shows that the $l_2$-norm of the difference of two random walk distributions $\p_v^t - \p_u^t$ is small for most pairs $u, v$ from the two different clusters if $t$ is not too large. For any vector $\p$ and vertex set $S$, let $\p(S):=\sum_{v\in S}\p(v)$.

\begin{proof}[Proof of Lemma \ref{lem:largel2}]
For any given subset $A \subseteq V$, vertex $v \in A$, and integer $t$, let $\rem(v,t,A)$ be the event that the lazy random walk of length $t$ starting at vertex $v$ never leaves $A$ in all $t$ steps. Let $I_A$ be the diagonal matrix such that $I_A(v,v)=1$ if $v\in A$ and $0$ otherwise. Then the probability that the walk stays entirely in $A$ is $(\1_v(\W_\reg I_A)^t)(A)$, that is, $\Pr[\rem(v,t,A)] = (\1_v(\W_\reg I_A)^t)(A)$. We will use the following claim.

\begin{claim}[Proposition 2.5 in \cite{ST13:local}]
\label{proposition-2.5-ST13:local}
For any $t\ge 1$ and any subset $A \subseteq V$ such that $\phi_G(A) \le \psi$, we have $\frac{\1_v(\W_\reg I_A)^t(A)}{|A|} \ge 1 - t \phi_G(A)/2 \ge 1 - t \psi/2$.
\end{claim}

Let $Q_A = \{v : \Pr[\rem(v,t,A)] \le 1 - \frac{t \psi}{2 \alpha}\}$. Then,
\begin{displaymath}
    1-\frac{\1_A}{|A|}(\W_\reg I_A)^t(A)
        =
    \sum_{v \in A}\frac{1}{|A|}(1 - \1_v(\W_\reg I_A)^{t}(A))
        \ge
    \sum_{v \in Q_A}\frac{1}{|A|}(1 - \1_v(\W_\reg I_A)^{t}(A))
        \ge
    \frac{|Q_A|}{|A|}\frac{t \psi}{2 \alpha}
    \enspace.
\end{displaymath}
From Claim \ref{proposition-2.5-ST13:local} and the inequality above, we conclude that $|Q_A| \le \alpha |A|$. Therefore, if we set $\widehat{A} = A\setminus Q_A$, then $|\widehat{A}| \ge (1-\alpha) |A|$, and for any $v \in \widehat{A}$, $\Pr[\rem(v,t,A)] \ge 1-\frac{t \psi}{2 \alpha}$. This proves the first part of the lemma.

To prove the second claim, we continue similarly and set $Q_B = \{v : \Pr[\rem(v,t,B)] \le 1 - \frac{t \psi}{2 \alpha}\}$ and define $\widehat{B} = B \setminus Q_B$, to obtain that $|\widehat{B}| \ge (1-\alpha) |B|$, and for any $v \in \widehat{B}$, $\Pr[\rem(v,t,B)] \ge 1-\frac{t \psi}{2 \alpha}$. Hence, for any $t \ge 1$ and $0 < \alpha < 1$, for any $u \in \widehat{A}$ and $v \in \widehat{B}$:%\Artur{Would the notation $\p_u^t(A)$ be clear for a reader? Or should we mention it somehow? {\bf Pan}: It was defined before, and now I added above "(Recall that .. )".}
\begin{displaymath}
    \p_u^t(A) \ge \Pr[\rem(u,t,A)] \ge 1-\frac{t\psi}{2\alpha}
        \quad \text{ and } \quad
    \p_v^t(B) \ge \Pr[\rem(v,t,B)] \ge 1-\frac{t\psi}{2\alpha}
    \enspace.
\end{displaymath}

Since $A$ and $B$ are disjoint, we have $\p_v^t(A) \le \p_v^t(V \setminus B) = 1 - \p_v^t(B) \le \frac{t \psi}{2 \alpha}$. Therefore, for any $t \ge 1$,
\begin{eqnarray*}
    \norm{\p_u^t-\p_v^t}_2
        & \ge &
    \frac{\norm{\p_u^t-\p_v^t}_1}{\sqrt{n}}
        =
    \frac{2 \max_{R \subseteq V} |\p_u^t(R)-\p_v^t(R)|}{\sqrt{n}}
        \ge
    \frac{2 (\p_u^t(A)-\p_v^t(A))}{\sqrt{n}}
        \\
        & \ge &
    \frac{2 (1 - \frac{t \psi}{2 \alpha} -\frac{t \psi}{2 \alpha}) }{\sqrt{n}}
        =
    \frac{2 (1 - \frac{t \psi}{ \alpha})}{\sqrt{n}}
    \enspace.
\end{eqnarray*}

In particular, if $t \le \frac{\alpha}{2\psi}$, then $\norm{\p_u^t - \p_v^t}_2 \ge \frac{1}{\sqrt{n}}$ and therefore $\norm{\p_u^t - \p_v^t}^2_2 \ge \frac{1}{n}$.
\end{proof}

\begin{remark}
It would be tempting to use in the above proof a somewhat stronger version of Claim \ref{proposition-2.5-ST13:local} that lower bounds the escaping probability by $\Omega(1)\cdot (1-3\psi/2)^t$ (see, for example,  \cite[Proposition 3.1]{OT12:local}). However, in our proof we we require the fraction of vertices in $\widehat{A}$ to be as large as $1 - \alpha$ for any small $\alpha > 0$, which we are not aware if it is true %\Artur{Is it ``cannot be'' or that ``we are not aware if it's true''?}
in the stronger version of Claim~\ref{proposition-2.5-ST13:local}.
\end{remark}

%\begin{remark}
%We note that the above upper bound is almost tight up to constants. Example: a graph that is composed of two disconnected $\phi_{in}$-expanders, each of size $n/2$.\Artur{To be done. I guess, this is not the best place for this remark; it should be somewhere earlier.}
%\end{remark}

%===========================================================================

\subsection{Partitioning into large sets with small cuts: Proof of Lemma \ref{lemma:partition-eps-far-improved}} \label{sec:epsfarlemma}

\newcommand{\siz}{h}%\ensuremath{\mathfrak{r}}}

%===========================================================================

In this section, we assume that $\varepsilon \le \frac12$ and we prove Lemma \ref{lemma:partition-eps-far-improved} that asserts that if a graph is far from $k$-clusterable then its vertex set can be partitioned into $k+1$ sets with low outer conductance. %We actually prove a stronger result that considers graphs that are $\varepsilon$-far from $(k, \phi^*_{in}, \phi^*_{out})$-clusterable for arbitrary $\phi^*_{out}$.
Let $0 < c_{\exp} \le \frac12$ be a constant such that for $d=3$ and every $n$, there exists
a graph $H$ with $n$ vertices and maximum degree $d=3$ that has $\phi(H) \ge c_{\exp}$. The proof of the next lemma follows the ideas from \cite{CS10:expansion}, but it is adapted to edge expansion and works also for $d=3$ (the analysis in \cite{CS10:expansion} requires $d \ge 4$).
%\Artur{I changed all (I hope) uses of $A - B$ in this section to $A \setminus B$, as it was used throughout the rest of the paper. Please stay consistent with this notation.}

\begin{lemma}
\label{lemma:subset1}
Let $\alpha \le \frac{c_{\exp}}{150 d}$. If for a graph $G=(V,E)$ there is $A\subseteq V$ with $|A| \le \frac19 \varepsilon |V|$ such that $\phi(G[V \setminus A]) \ge c_{\ref{lemma:subset1}} \cdot \alpha$ for some sufficiently large constant $c_{\ref{lemma:subset1}}$, then $G$ is not $\varepsilon$-far from every graph $H$ with $\phi(H) \ge \alpha$.
\end{lemma}

\begin{proof}
Let $c_{\ref{lemma:subset1}}$ be a sufficiently large constant whose value will be determined later. Let $G$ be a graph as in the lemma and let $A \subseteq V$ be an arbitrary set such that $A \subseteq V$ with $|A| \le \frac19 \varepsilon |V|$ and $\phi(G[V \setminus A]) \ge c_{\ref{lemma:subset1}} \cdot \alpha$. We will turn $G$ into a graph $H$ by modifying at most $\varepsilon d n$ edges of $G$ and then prove that $\phi(H) \ge \alpha$. This will conclude the proof.

Our construction removes all edges between vertices in $A$ and adds an expander graph with maximum degree $3$ on $A$ that has a constant fraction of vertices of degree $2$. The degree $2$ vertices are
then connected to vertices $V\setminus A$. In order to not violate the degree bound, we have to remove some edges between vertices in $V\setminus A$, which is done using the following construction.

We will first construct an auxiliary set $S$ of size $\lceil |A|/4 \rceil$. Each element of set $S$ is an edge $\{u,v\}$ for some $u, v \in V \setminus A$ (we allow selfloops).
The set $S$ can be constructed by the following algorithm.%\Artur{Why do we have ``if $d_G(v) \ge d-1$ then'' in the second loop? Didn't we already remove from $U$ all vertices with $d_G(v) < d-1$ in the first for-loop?}
\begin{center}
\begin{tabular}{|p{0.5\textwidth}|}
\hline
{\sc ConstructS}($G,A)$\\
\hline\\[-0.35in]
\begin{tabbing}
\hspace{0.5cm}\= $Q_L = \{ u \in V\setminus A: d_G(u) \le d-2 \}$ \\
\> $S' = \{\{v,v\}: v \in Q_L\}$\\
\> $U = (V \setminus A) \setminus Q_L$\\
\> \textbf{while} there is $v \in U$ with at least one neighbor in $U$ \textbf{do}\\
\>\hspace{0.5cm}\= let $u \in U$ be a neighbor of $v$\\
\>\> $S' = S' \cup \{\{u,v\}\}$\\
\>\> $U = U \setminus \{u,v\}$\\
\> \textbf{return} set $S$ defined as an arbitrary subset of $S'$ of size $\lceil |A|/4 \rceil$
\end{tabbing}\\
\hline
\end{tabular}
\end{center}
\junk{
\begin{center}
\begin{tabular}{|p{0.5\textwidth}|}
\hline
{\sc ConstructS}($G,A)$\\
\hline\\[-0.35in]
\begin{tabbing}
\hspace{0.5cm}\= $S = \emptyset$\\
\> $U = V-A$\\
\> {\bf for each } $v \in U$ {\bf do} \\
\> \hspace{0.5cm} \= {\bf} if $d_G(v)\le d-2$ {\bf then }\\
\>\>\hspace{0.5cm}\= $ S = S \cup \{\{v,v\}\}$\\
\>\>\> $U = U - \{v\}$\\
\> {\bf for each } $v \in U$ {\bf do}\\
\>\> {\bf if} $d_G(v) \ge d-1$ {\bf then }\\
\>\>\> {\bf if} $v$ has at least one neighbor in $U$ {\bf then } \\
\>\>\>\hspace{0.5cm}\= Let $u$ be such a neighbor of $v$\\
\>\>\>\> $S=S \cup \{\{u,v\}\}$\\
\>\>\>\> $U= U - \{u,v\}$
\end{tabbing}\\
\hline
\end{tabular}
\end{center}
}

We prove that {\sc ConstructS} ensures that $|S'| \ge \frac{1}{6} |V|$, which implies that the last step of the algorithm can always be executed and we get $|S| = \lceil |A|/4 \rceil$.

\begin{claim}
If algorithm {\sc ConstructS} is invoked with $A$ that satisfies $|A| \le \frac19 \varepsilon |V|$, $0 < \varepsilon \le \frac12$,
then the constructed set $S'$ has size at least $\frac16 |V|$.
\end{claim}

\begin{proof}
We first observe that at the end of the algorithm, each vertex in $U$ has degree at least $d-1$ and all the neighbors of vertices in $U$ belong to $V \setminus U$. This implies that the number of edges connecting $U$ and $V \setminus U$ is on one hand, at least $(d-1)|U|$, and on the other hand, it is at most $d |V \setminus U|$. Therefore, $d |V \setminus U| \ge (d-1)|U|$, and since $d \ge 3$, this yields $|V \setminus U| \ge \frac23 |U|$, and thus $|U| \le \frac35 |V|$.

Now, we observe that $|S'| \ge \frac12 |(V \setminus A) \setminus U|$, and therefore $|S'| \ge \frac12 (|V| - |A| - |U|) \ge \frac12 (|V| - \frac{1}{18} |V| - \frac35|V|) = \frac{31}{180} |V| \ge \frac16 |V|$, for every $A$ that satisfies the prerequisites of the claim.
%that $|A| \le \frac19 \varepsilon |V|$ and $\varepsilon \le \frac12$.
\end{proof}

We next describe our construction of the graph $H$. If $|A| \ge 10$, then we proceed as follows. We partition $A$ into two sets $A'$ and $A''$, with $|A''| = 2 \cdot  \lceil |A|/4 \rceil$. Let $H' =(A', E')$ be a graph with degree at most $3$ and %\Artur{I'm assuming it was a typo and instead of $\phi(H) \ge c_{\exp}$ I put $\phi(H') \ge c_{\exp}$. Please confirm.}
$\phi(H') \ge c_{\exp}$, whose existence follows from our definition of $c_{\exp}$. Since adding edges (while maintaining the degree bound) does not decrease the conductance and since $|A| \ge 10$, we may assume that $H'$ has at least $|A''|$ edges. Let $H^* = (A,E^*)$ be a graph obtained from $H'$ by taking an arbitrary set of $|A''|$ edges from $E'$ and replacing them by a path of length two, whose intermediate vertex is from $A''$ in such a way that every vertex from $A''$ is used exactly once.

If $1< |A| < 10$ we define $H^*=(A,E^*)$ to be a path and choose $A''$ to be an arbitrary subset of $A$ of size $2 \lceil \frac{|A|}{4} \rceil$. If $|A|=1$ we define $H^* = (A,E^*)$ with $E^* = \emptyset$, and set $A' =\emptyset$ and $A'' = A$.
%\Artur{Is it OK now?}

Now we will modify $G$ by changing at most $\varepsilon d n$ edges to construct graph $H$ such that $\phi(H) \ge \alpha$. We first remove in $G$ all edges incident to $A$ and then all edges that connect the sets $s \in S$ in $G$ (i.e., we remove from $E$ all edges $(u,v)$ with $u, v \in s$). Then we add an arbitrary perfect matching between the vertices in $A''$ and $S$ (if a vertex appears twice in $s \in S$ then it will be matched to two vertices of $A''$; if $|A''|=1$, then the vertex $v$ from $A''$ will be match to both vertices from $s \in S$. If, in this case, $s=(u,u)$ we only add the edge $(u,v)$).
%\Artur{I'm not sure if I follow this. Don't we always have the size of $A''$ to be even? Isn't it either $2 \lfloor \frac{|A|}{4} \rfloor$ or $2 \lceil \frac{|A|}{4} \rceil$?And from what I can see, $A''$ might look undefined if $|A|=1$.}
Finally, we add all edges $E^*$ from the graph $H^*$ defined above.

Our construction creates a new graph $H$ from $G$ by making at most $(d+1)|A|$ edge deletions and $3|A|$ edge insertions. Hence, we modified at most $(d+4) |A| \le \varepsilon d|V|$ edges, as required.

Next we prove that $\phi(H) \ge \alpha$. We begin with two auxiliary claims about construction of $H$.

%\Artur{In the analysis below you missed factor $d$ in many places (since you're claiming that if $\phi(Q) \ge \rho$ for some graph $Q = (U,E_U)$, then $e_Q(Z, U \setminus Z) \ge \rho \min\{|Z|, |U \setminus Z|\}$ for any subset $Z \subseteq U$, whereas you should have obtained $e_Q(Z, U \setminus Z) \ge \rho d \min\{|Z|, |U \setminus Z|\}$). However, since this would made the bounds stronger rather than weaker (and so the bounds still hold), I decided to not change it to make sure I'm not introducing any new errors.}

%\Artur{In the analysis below you're very often giving weak estimations. While all this is true, it's very confusing for the reader (for example, the bound in Claim \ref{claim:edges1}). This is also very confusing for me, or any checker, since I don't know if this is coming from your calculations which I maybe didn't understand (and hence I think it should be A whereas you're putting B), or this is just so non-important that you live with weaker bounds.}

\begin{claim}
\label{claim:edges1}
Let $X \subseteq V$ be an arbitrary set of size at most $\frac12 |V|$. Then the following holds:
%\Artur{I'm quite sure we want to have $e_H(X, V \setminus X) \ge \tfrac{1}{15} c_{\exp} \cdot d \cdot \min\{|X \cap  A| , |A \setminus X|\}$, that is, \textbf{with} factor $d$. C.S.: No, at this place we cannot get a $d$.}
%
\begin{displaymath}
    e_H(X, V \setminus X)
        \ge
    \frac{1}{15} c_{\exp} \cdot \min\{|X \cap  A| , |A \setminus X|\}
    \enspace.
\end{displaymath}
\end{claim}

\begin{proof}
If $|A| = 1$ the claim trivially holds for every set $X$. Thus, we can assume $|A| \ge 2$. Let $X$ be a subset of $V$ of size at most $\frac 12 |V|$. If $|A| < 10$, we get $e_H(X, V \setminus X) \ge e_H(X \cap A, A \setminus X) \ge \frac{1}{10} \cdot \min\{|X \cap  A| , |A \setminus X|\}$, since either the minimum is $0$ or there is at least one edge connecting the two sets. Since $c_{exp}\le \frac12$, this implies the claim.

Now we consider the case $|A| \ge 10$. Consider an arbitrary set $Y \subseteq A$ with $|Y| \le \frac12 |A|$. Let $Y' = Y \cap A'$ and $Y'' = Y \cap A''$. Let us first focus on the construction of graph $H^*$ (which is a subgraph of $H$). Let $Y^* \subseteq Y''$ be the set of vertices from $Y''$ with both of its neighbors (in $H^*$) to be in $Y$ (and hence, in fact, in $Y' \subseteq A'$).
%\Artur{What about making a figure here? Showing sets $Y', Y'', Y*$, etc}

We consider two cases. If $|Y'' \setminus Y^*| \ge \frac12 |Y|$ then since each vertex in $Y'' \setminus Y^*$ is adjacent in $H^*$ to at least one vertex not in $Y$, we obtain $e_{H^*}(Y, A \setminus Y) \ge |Y'' \setminus Y^*| \ge \frac12 |Y|$.

Otherwise we have $|Y'' \setminus Y^*| < \frac12 |Y|$, and thus $|Y'| + |Y^*| > \frac12 |Y|$. Since each vertex in $Y'$ has degree at most $3$ in $H^*$ and each vertex in $Y^*$ is adjacent in $H^*$ to exactly two vertices from $Y'$, we have $|Y^*| \le \frac32 |Y'|$. Hence, if we combine the bounds $|Y'| + |Y^*| > \frac12 |Y|$ and $|Y^*| \le \frac32 |Y'|$, then we obtain $|Y'| > \frac15 |Y|$.
Now we make another case distinction.

If $|Y'| \le \frac{9}{10} |A'|$, then $|A' \setminus Y'| \ge \frac{1}{10} |A'| \ge \frac 19 |Y'|$. Note that in our construction of $H^*$ from $H'$, if an edge $(u,v)$ with $u \in A' \setminus Y'$ and $v\in Y'$ is replaced by a path of length $2$ with intermediate vertex $w\in A''$, then at least one of the edges $(u,w)$ and $(v,w)$ lies between $Y$ and $A \setminus Y$ in $H^*$. Therefore, %\Artur{This is confusing. Later we don't use $e_H(Y, V \setminus Y) \ge \tfrac{1}{15} c_{exp} |Y|$ but only that $e_{H^*}(Y, A\setminus Y) \ge \tfrac{1}{15} c_{exp} |Y|$. Would you want to have
%
%\begin{displaymath}
%    e_{H^*}(Y, A \setminus Y)
%        \ge
%    e_{H'}(Y', A \setminus Y')
%        \ge
%    3 c_{exp} \min\{|Y'|, |A' \setminus Y'|\}
%        \ge
%    \tfrac13 c_{exp} |Y'|
%        \ge
%    \tfrac{1}{15} c_{exp} |Y|
%        \enspace.
%\end{displaymath}
%? If so, then please add comments about the first inequality.}
%
\begin{displaymath}
    e_{H^*}(Y, A \setminus Y)
        \ge
    e_{H'}(Y', A \setminus Y')
        \ge
    3 c_{exp} \min\{|Y'|, |A' \setminus Y'|\}
        \ge
    \tfrac13 c_{exp} |Y'|
        \ge
    \tfrac{1}{15} c_{exp} |Y|
        \enspace.
\end{displaymath}

Otherwise, $|Y'| \ge \frac{9}{10} |A'|$.
%$ \ge \frac{9}{25} |A|$, where the last inequality follows from $|A'|\ge \frac 25 |A|$.
In our construction we replace $2 \cdot \lceil |A|/4 \rceil$ edges of $H'$ by paths of length $2$. Since $|A' \setminus Y'| \le \frac {1}{10} |A'| \le\frac {1}{10} |A| $ and since $H'$ has maximum degree $3$, there are at most $\frac{3}{20} |A|$ edges with both endpoints in $A'\setminus Y'$ that are replaced. Therefore, there are $2 \lceil |A|/4 \rceil - \frac{3}{20} |A| \ge \frac{7}{20} |A|$ edges replaced that in $H'$ are incident to a vertex from $Y'$. Thus, in $H^*$ there are at least $\frac{7}{20} |A|$ edges leaving $Y'$.
Since $|Y'| \ge \frac{9}{10} |A'|$ and %\Artur{Christian wrote earlier that $|A'| \ge \frac 25 |A|$, but from what I can see, we have (in the case $|A| \ge 10$, which is what we're considering now) $|A'| \ge \frac 12 |A|$. {\bf Pan:} No, we only get $|A'| \ge \frac 25 |A|$, since $|A''|=2 \cdot \lceil |A|/4 \rceil$. (There was a typo $|A''|=2 \cdot \lfloor |A|/4 \rfloor$ before) }
$|A'| \ge \frac 25 |A|$, we have $|Y'| \ge \frac{9}{25} |A|$.
Therefore our assumption that $|Y| \le \frac12 |A|$ yields $|Y \setminus Y'| \le \frac{7}{50} |A|$. This gives us %\Artur{Is this $e(Y, V \setminus Y) = e_G(Y, V \setminus Y)$ or $e_H(Y, V \setminus Y)$ or $e_{H^*}(Y, V \setminus Y)$?}
$e_{H^*}(Y, A \setminus Y) \ge \frac{7}{20} |A| - \frac{7}{50} |A| = \frac{21}{100} |A| \ge \frac{1}{5} |A|\ge \frac{1}{5} |Y| \ge \frac{1}{15} c_{\exp} |Y|$.

%\Artur{I'm confused now. Is this bound $e_{H^*}(Y, A \setminus Y) \ge \frac{1}{15} c_{\exp} \cdot |Y|$ for the case $|Y'' \setminus Y^*| < \frac12 |Y|$, or for the case $|Y'| \ge \frac{9}{10} |A'|$, or maybe for all cases?}
Therefore, we get $e_{H^*}(Y, A \setminus Y) \ge \frac{1}{15} c_{\exp} \cdot |Y|$ for the case that $|Y'' \setminus Y^*| < \frac12 |Y|$.

If we combine the bounds for these two cases together, then we obtain that for any $Y \subseteq A$ with $|Y| \le \frac12 |A|$, we have $e_{H^*}(Y, A \setminus Y) \ge \min\{\frac12 |Y|, \frac{1}{15} c_{\exp} |Y|\} = \frac{1}{15} c_{\exp} |Y|$. This further implies that for any $Y \subseteq A$, $e_{H^*}(Y, A \setminus Y) \ge \frac{1}{15} c_{\exp}\min\{|Y|, |A \setminus Y|\}$. %, where the last identity follows from our assumption that $|Y| \le \frac12 |A|$.

Now we will extend the analysis to the graph $H$. We have $e_H(X, V \setminus X) \ge e_{H^*}(X \cap A, A \setminus X) \ge \frac{1}{15} c_{\exp} \min\{|X \cap A|, |A \setminus X|\}$.
\end{proof}

\begin{claim}
\label{claim:edges2}
Let $X \subseteq V$ be an arbitrary set of size at most $\frac12 |V|$, $A \subseteq V$ with $|A| \le \frac19 \varepsilon |V|$ and $\varepsilon \le \frac12$. Then the following holds:%\Artur{The analysis gives $e_H(X, V \setminus X) \ge \tfrac89 \cdot c_{\ref{lemma:subset1}} \cdot d \cdot \alpha \cdot |(V \setminus A) \cap X| - \min\{|X \cap A|, |A \setminus X|\}$, but since I don't think it matters much, I left the old analysis and just in one place (last sentence) I wrote $\frac89 \ge \frac45$.}
%\Artur{You introduced notion $B = V \setminus A$ here, and then used it in the proof and later in 3-4 places; I found it a poor idea; it doesn't increase the readability, but only introduce yet another term/variable. But to not change too much, I left it in the proof; perhaps one could remove it there as well.}
%\Artur{Again, I would expect factor $d$ missing, especially since the proof below already gives $e_H(X, V \setminus X) \ge \tfrac45 \cdot c_{\ref{lemma:subset1}} \cdot \alpha \cdot d \cdot |(V \setminus A) \cap X| - \min\{|X \cap A|, |A \setminus X|\}$.}
%
\begin{displaymath}
    e_H(X, V \setminus X)
        \ge
    \tfrac45 \cdot c_{\ref{lemma:subset1}} \cdot d \cdot \alpha \cdot |(V \setminus A) \cap X| -
        \min\{|X \cap A|, |A \setminus X|\}
    \enspace.
\end{displaymath}
\end{claim}

\begin{proof}
For simplicity of notation, let us define $B = V \setminus A$.  Using the assumption $|A| \le \frac19 \varepsilon |V|$ and $\varepsilon \le \frac12$, we obtain $|B| \ge (1 - \frac19 \varepsilon) |V| \ge \frac{17}{18} |V|$. Therefore, since $|B \cap X| \le |X| \le \frac12 |V|$, we obtain $|B \cap X| \le \frac{9}{17} \cdot |B|$, and hence $|B \setminus X| = |B| - |B \cap X| \ge \frac{8}{17} \cdot |B|$, what yields $\min\{|B \cap X|, |B \setminus X|\} \ge \frac89 |B \cap X|$.
Next, by the assumption about set $A$ in Lemma \ref{lemma:subset1}, we know that $\phi(G[B]) \ge c_{\ref{lemma:subset1}} \cdot \alpha$. Therefore, $e_{G[B]}(B \cap X, B \setminus X) \ge c_{\ref{lemma:subset1}} \alpha d \min\{|B \cap X|, |B \setminus X|\} \ge \frac89 c_{\ref{lemma:subset1}} \alpha d |B \cap X|$.

The only edges that are removed from $G[B]$ in order to obtain $H$ are the edges between vertices $u,v$ with $u,v \in s$ for all $s \in S$. Consider such an edge $(u,v)$ with $u, v \in s$, $s \in S$. Since we are analysing the size of the cut between $B \cap X$ and $B \setminus X$, we only consider $u \in B \cap X$ and $v \in B \setminus X$. By our construction of $H$, both $u$ and $v$ are connected in $H$ to vertices in $A$. If $u$ is connected to a vertex in $A \setminus X$ or $v$ to a vertex in $A \cap X$, then we get a new cut edge between $B \cap X$ and $B \setminus X$, and thus this will compensate the removal of edge $(u,v)$ from $G[B]$. Therefore, we decrease the number of edges in the cut between $B \cap X$ and $B \setminus X$ only if $u$ is connected to a vertex in $A \cap X$ and $v$ is connected to a vertex in $A \setminus X$. Each vertex in $A$ is adjacent in $H$ to at most one vertex from outside $A$, and therefore the number of such edges is bounded by $\min\{|X \cap A|, |A \setminus X|\}$.

If we summarize this, we obtain $e_H(X, V \setminus X) \ge e_{G[B]}(B \cap X, B \setminus X) - \min\{|X \cap A|, |A \setminus X|\} \ge \frac89 c_{\ref{lemma:subset1}} \alpha d |B \cap X| - \min\{|X \cap A|, |A \setminus X|\} \ge \frac45 c_{\ref{lemma:subset1}} \alpha d |B \cap X| - \min\{|X \cap A|, |A \setminus X|\}$.
\end{proof}

With Claims \ref{claim:edges1} and \ref{claim:edges2} at hand, we are ready to conclude the proof of Lemma \ref{lemma:subset1}. Take an arbitrary set $X \subseteq V$ of size at most $\frac12 |V|$. We will prove that $e_H(X,V \setminus X) \ge \alpha d |X|$, what would immediately imply that $\phi(H) \ge \alpha$.%\Artur{Of course, this is \textbf{NOT true}, since we're missing factor $d$. Someone please carefully fix it and make sure that $e_H(X,V \setminus X) \ge \alpha d |X|$.}

If $\min\{|X \cap A|, |A \setminus X|\} \ge \frac {15 \cdot d \cdot \alpha}{c_{\exp}} \cdot |X|$, then Claim \ref{claim:edges1} gives that $e_H(X,V \setminus X) \ge \alpha d |X|$. Otherwise, we have $\min\{|X \cap A|, |A \setminus X|\} < \frac{15 \cdot d \cdot \alpha}{c_{\exp}} \cdot |X| \le \frac{1}{10} \cdot |X|$ for our choice of $\alpha$. If the minimum is attained by $|X \cap A|$, then we have $|(V \setminus A) \cap X| \ge \frac{9}{10} \cdot |X|$. Thus Claim \ref{claim:edges2} implies that assuming that $c_{\ref{lemma:subset1}} \ge \frac{30}{c_{\exp}}$, we have $e_H(X,V \setminus X) \ge \frac45 \cdot d \cdot  c_{\ref{lemma:subset1}} \cdot \alpha \cdot \frac{9|X|}{10} - \frac{15 \cdot d \cdot  \alpha}{c_{\exp}}|X| \ge \alpha d\cdot |X|$.

If the minimum is attained by $|A \setminus X|$ we consider two cases. If $|(V\setminus A) \cap X| \le \frac{1}{16} |A|$ then $|X| \le |A| + |(V\setminus A) \cap X| \le \frac{17}{16} |A|$.
In this case, $|X\cap A| = | A \setminus (A\setminus X)| = |A| - |A \setminus X| \ge |A| - |X|/10 \ge \frac{143}{160} |A| \ge \frac45 |A|$.
Since $|A''| \ge \frac12 |A|$ we obtain that $|X \cap A''| \ge |X\cap A| - |A'| \ge \frac45 |A| - \frac12 |A| \ge \frac{3}{10} |A|$.
By construction of $H$ each vertex in $A''$ is connected to a vertex in $V \setminus A$ and each vertex in $V\setminus A$ is connected to at most $2$ vertices in $A''$. Since
$|(V\setminus A) \cap X| \le \frac{1}{16} |A|$ there are at most $\frac18 |A|$ vertices of $(V \setminus A) \cap X$ connected to vertices from $X\cap A''$. Hence, for our choice of $\alpha$
there are at least $\frac{3}{10} |A| - \frac18 |A| \ge \frac{1}{10} |A| \ge \frac{16}{170} |X| \ge \alpha d |X|$ edges leaving $X$.

%If the minimum is attained by $|A \setminus X|$, then the number of vertices of $X \cap A$ that are also in $A''$ is at least\Artur{Why? Explain!} $\frac14 |A| - \frac{1}{10} |A| \ge \frac18 |A|$. We consider now two cases. If $|(V \setminus A) \cap X| \le \frac{1}{16} |A|$, then $|X| \le |A| + |(V \setminus A) \cap X| \le \frac{17}{16} |A|$, and we have\Artur{Why do you have the first bound, $e_H(X,V \setminus X) \ge \frac{1}{16} |A|$? Explain!}\Artur{Why can you assume that $\alpha \le \frac{1}{17}$?} $e_H(X,V \setminus X) \ge \frac{1}{16} |A| \ge \frac{1}{17} |X| \ge \alpha \cdot |X|$.
If $|(V \setminus A) \cap X| > \frac{1}{16} |A|$, then $|(V \setminus A) \cap X| > \frac{1}{16} |A \cap X|$ and thus $|(V \setminus A) \cap X| > \frac{1}{17} |X|$.
Now if $c_{\ref{lemma:subset1}} \ge \frac{350}{c_{\exp}}$, Claim \ref{claim:edges2} gives that $e_H(X,V \setminus X) \ge \frac45 \cdot d\cdot  c_{\ref{lemma:subset1}} \cdot \alpha \cdot \frac{|X|}{17} - \frac{15 \cdot \alpha d}{c_{\exp}}|X| \ge \alpha d |X|$. Therefore, Lemma \ref{lemma:subset1} follows with $c_{\ref{lemma:subset1}} = \frac{350}{c_{\exp}}$.
%\Artur{Honestly, I gave it up here. Lack of explanations made it very hard to read. I would hope the bounds are correct, but it's impossible to say it 100\% without checking every detail, and for that I would need you to answer to all my queries first, and in fact, I would like to see some comments in the last few lines of this proof.}
\end{proof}

Lemma \ref{lemma:subset1} can be applied to construct a large set $A$ with a small cut, as in the following lemma.

\begin{lemma}
\label{lemma:subset}
Let $0 < \alpha \le \frac{c_{\exp}}{150d}$ and $0 < \varepsilon \le \frac12$. If $G = (V,E)$ is $\varepsilon$-far from any graph $H$ with $\phi(H) \ge \alpha$, then there is a subset of vertices
%\Artur{The analysis (as I can see it) gives only $\frac{1}{18} \varepsilon |V| \le |A|$. This should be corrected either in the statement (and then in the use of Lemma \ref{lemma:subset}), or one should give a proof with $\frac19 \varepsilon |V| \le |A|$ (right now we could easily get $\frac19 \varepsilon |V| \le |A| \le (\frac12 + \frac{\varepsilon}{9}) |V|$ by taking always $A = A_1 \cup \dots \cup A_i$).}
$A \subseteq V$ with $\frac{1}{18} \varepsilon |V| \le |A| \le \frac12 |V|$ such that $\phi_G(A) \le c_{\ref{lemma:subset}} \cdot \alpha$, for some sufficiently large constant $c_{\ref{lemma:subset}}$.
In particular, $e(A, V \setminus A) \le c_{\ref{lemma:subset}} \cdot \alpha \cdot d \cdot |A|$. %\Artur{We didn't have $d$ before and we had dependency on $|V|$ before, now it's on $|A|$.}
\end{lemma}

\begin{proof}
Lemma \ref{lemma:subset1} ensures that if $G$ is $\varepsilon$-far from any graph $H$ with $\phi(H) \ge \alpha$, then for all $A' \subseteq V$ with $|A'| \le \frac19 \varepsilon |V|$ we have $\phi(G[V \setminus A']) < c_{\ref{lemma:subset1}} \cdot \alpha$. In particular, in our case, this will mean that there is a set $B \subseteq V \setminus A'$ with $|B| \le \frac12 |V \setminus A'|$ such that
%$\phi_{G[V \setminus A']}(B) < c_{\ref{lemma:subset1}} \cdot \alpha$.
$e(B, (V \setminus (A' \cup B)) < c_{\ref{lemma:subset1}} \cdot \alpha \cdot d \cdot |B|$.

We will now repeatedly apply Lemma \ref{lemma:subset1} to construct a large set $A$ satisfying the requirements of Lemma \ref{lemma:subset}. Let $A_1 = \emptyset$. We apply Lemma \ref{lemma:subset1} with $A' = A_1$ to obtain a set $A_2$ with $|A_2| \le \frac12 |V \setminus A'|$ and $\phi_{G[V \setminus A']}(A_2) \le c_{\ref{lemma:subset1}} \cdot \alpha$. If $|A_1 \cup A_2| \ge \frac19 \varepsilon |V|$ then we are done. Otherwise, we set $A' = A_1 \cup A_2$ and repeat this process. We continue this process until for the first time, we obtain a set $A_i$ such that $|A_1 \cup \dots \cup A_i| \ge \frac19 \varepsilon |V|$. In that moment, if $|A_i| \ge |A_1 \cup \cdots \cup A_{i-1}|$ then we set $A = A_i$ and otherwise, we put $A = A_1 \cup \dots \cup A_i$.

Our construction ensures that since $|A_1 \cup \dots \cup A_i| \ge \frac19 \varepsilon |V|$, then we have $|A| \ge \frac{1}{18} \varepsilon |V|$.
%\Artur{Alas, we're getting only $|A| \ge \frac{1}{18} \varepsilon |V|$ rather than $|A| \ge \frac{1}{9} \varepsilon |V|$. This is because we may have that $|A_1 \cup \cdots \cup A_{i-1}| = \frac{1}{18} \varepsilon |V|$ and $|A_i| = \frac{1}{18} \varepsilon |V|$, in which case $|A_1 \cup \cdots \cup A_i| \ge \frac19 \varepsilon |V|$ and $|A_i| \ge |A_1 \cup \cdots \cup A_{i-1}|$, and therefore after setting $A = A_i$, we will have $|A| = \frac{1}{18} \varepsilon |V|$.}
The upper bound on the size of $A$ follows since $|A_i| \le \frac12 |V|$ and $|A_1 \cup \dots \cup A_{i-1}| < \frac19 \varepsilon |V|$.

Our construction ensures that for every $1 \le j \le i$, $e(A_j, V \setminus (A_1 \cup \dots \cup A_j)) \le c_{\ref{lemma:subset1}} \cdot \alpha \cdot d \cdot |A_j|$. Therefore, since we have $e(A_1 \cup \dots \cup A_j, V \setminus (A_1 \cup \dots \cup A_j)) \le \sum_{s=1}^j e(A_s, V \setminus (A_1 \cup \dots \cup A_s))$, we conclude that $e(A_1 \cup \dots \cup A_j, V \setminus (A_1 \cup \dots \cup A_j)) \le c_{\ref{lemma:subset1}} \cdot \alpha \cdot d \cdot |A_1 \cup \dots \cup A_j|$. Hence, if $A = A_1 \cup \dots \cup A_i$ then we obtain $e(A, V \setminus A) \le c_{\ref{lemma:subset1}} \cdot \alpha \cdot d \cdot |A|$, and if $A = A_i$ then we obtain
\begin{eqnarray*}
    e(A, V \setminus A)
        & = &
    e(A_i, A_1 \cup \dots \cup A_{i-1}) +
        e(A_i, V \setminus (A_1 \cup \dots \cup A_i))
        \\
        & \le &
    e(A_1 \cup \dots \cup A_{i-1}, V \setminus (A_1 \cup \dots \cup A_{i-1})) +
        e(A_i, V \setminus (A_1 \cup \dots \cup A_i))
        \\
        & \le &
    c_{\ref{lemma:subset1}} \cdot \alpha \cdot d \cdot |A_1 \cup \dots \cup A_{i-1}| +
        c_{\ref{lemma:subset1}} \cdot \alpha \cdot d \cdot |A_i|
        \\
        & \le &
    2 c_{\ref{lemma:subset1}} \cdot \alpha \cdot d \cdot |A|
    \enspace,
\end{eqnarray*}
where in the last inequality we use the fact that $|A| = |A_i| \ge |A_1 \cup \cdots \cup A_{i-1}|$.

This completes the proof by setting $c_{\ref{lemma:subset}} = 2 c_{\ref{lemma:subset1}}$.
\end{proof}

Let us extend the notion $e(U_1,U_2)$ to multiple sets and for disjoint subsets $V_1, \dots, V_h$, let us define $e(V_1, \dots, V_h) = \sum_{1 \le i < j \le h} e(V_i, V_j)$.

\begin{lemma}
\label{lemma:partition}
Let $G = (V,E)$ be $\varepsilon$-far from $(k,\phi_{in}^*,\phi_{out}^*)$-clusterable and $\phi_{in}^* \le c_{\exp}/d$. If there is a partition of $V$ into $h$ sets $V_1, \dots, V_h$ with $1 \le h \le k$, such that $e(V_1, \dots, V_h) = 0$, then there is an index $i$, $1 \le i \le h$, with $|V_i| \ge \frac{1}{8k} \cdot \varepsilon |V|$ such that $G[V_i]$ is $\frac{\varepsilon}{2}$-far from any $H$ on vertex set $V_i$ with maximum degree $d$ and $\phi(H) \ge \phi_{in}^*$.
\end{lemma}

\begin{proof}
Let us renumber the indices of sets $V_1, \dots, V_{\siz}$ such that $|V_i| \ge |V_{i+1}|$ for every $i$, $1 \le i < \siz$. A set $V_i$ with more than $\frac{1}{8k} \varepsilon |V|$ vertices is called \emph{large} and otherwise it is called \emph{small}. Let $s$ be the largest index such that $V_s$ is large. (Simple counting arguments implies that we must have $|V_1| \ge \frac{|V|}{\siz}$ (for otherwise we would have $|V_i| < \frac{|V|}{\siz}$ for every $i$, $1 \le i \le \ell$, and thus $\sum_{i=1}^{\siz}|V_i| < |V|$, which is a contradiction to the fact that $V_1, \dots, V_{\siz}$ is a partition of $V$), and hence $V_1$ is large and $s$ is well-defined.) Next, let us observe that $\sum_{1 \le i \le \siz: V_i \text{ is small}} |V_i| \le \frac{1}{8k} k \varepsilon |V| = \frac18 \varepsilon |V|$. This follows from $\siz \le k$ and from the fact that for a small set $V_i$ we have $|V_i| \le \frac{1}{8k} \varepsilon |V|$.

Let us construct from $G$ a new graph $G^*$ of maximum degree at most $d$ as follows.  Define $U = \bigcup_{i: V_i \text{ is small}} V_i = \bigcup_{i=s+1}^{\siz} V_i$, and remove in $G$ all edges incident to any vertex in $U$. Then build a degree $3$ \ $c_{\exp}$-expander on $U$ and add it to the graph. Note that with respect to $d$, this expander is a $\frac{c_{exp}}{d}$-expander. Call the obtained graph $G^*$.

Observe that $G^*$ has been obtained from $G$ by adding/inserting at most $ d |U| + 3 |U|$ edges, where the first term corresponds to the removal of all edges incident to $U$ and the second term corresponds to building the degree $3$ $c_{\exp}$-expander on $U$.

Now, since $|U| = \sum_{1 \le i \le \siz: V_i \text{ is small}} |V_i| \le \frac18 \varepsilon |V|$, as we have shown above, we note that $G^*$ is obtained from $G$ by adding/deleting at most $ 2\cdot \frac{d}{8} \varepsilon |V| \le \frac{d}{2} \varepsilon |V|$ edges.
%\Artur{Why not $\frac{d}{4} \varepsilon |V|$ edges?}.
Hence, since $G^*$ has maximum degree at most $d$, $G^*$ is $\frac12 \varepsilon$-far from $(k, \phi^*_{in}, \phi^*_{out})$-clusterable.

Observe the structure of $G^*$: it consists of a $\frac{c_{\exp}}{d}$-expander on $U$ and $s$ disjoint components (not necessarily connected) on vertex sets $V_i$ with each $V_i$ being a large set and $G^*[V_i] = G[V_i]$; further, $\phi_{G^*}(U) = \phi_{G^*}(V_1) = \dots = \phi_{G^*}(V_s) = 0$.

For every $i$, $1 \le i \le s$, let us define $H_i$ to be the graph on vertex set $V_i$ with maximum degree at most $d$, with $\phi(H_i) \ge \phi^*_{in}$, and that is obtained from $G^*[V_i]$ by the minimum number of addition/deletion of the edges; let $\kappa_i$ be the number of addition/deletion of the edges needed to transform $G^*[V_i]$ into $H_i$. %(While $H_i$ may be not unique, $\kappa_i$ is well-defined and is unique.)

Let us observe that the graph $H$ on $V$ obtained as the union of $G^*[U]$ and $H_1, \dots, H_s$ is $(k, \phi^*_{in}, \phi^*_{out})$-clusterable. Indeed, since we have $H[U] = G^*[U]$, $H[V_i] = H_i$ for every $i$, $1 \le i \le s$, and $\phi_H(U) = \phi_H(V_1) = \dots = \phi_H(V_s) = 0$, for the partition of $V$ into $U$, $V_1, \dots, V_s$, we obtain that $\phi(H[U]) \ge c_{\exp}/d \ge \phi^*_{in}$ for every $i$, $1 \le i \le s$, and $\phi_H(U) = \phi_H(V_1) = \dots = \phi_H(V_s) = 0 \le \phi^*_{out}$.

We now note that $H$ is obtained from $G^*$ by adding $\sum_{i=1}^s \kappa_i$ edges. Therefore, since $G^*$ is $\frac12 \varepsilon$-far from $(k, \phi^*_{in}, \phi^*_{out})$-clusterable, since $H$ is $(k, \phi^*_{in}, \phi^*_{out})$-clusterable, we must have $\sum_{i=1}^s \kappa_i > \frac12 \varepsilon d |V|$, and thus $\sum_{i=1}^s \kappa_i > \frac12 \varepsilon d \sum_{i=1}^s |V_i|$. Therefore, there must be at least one $j$, $1 \le j \le s$, with $\kappa_j > \frac12 \varepsilon d |V_j|$. In that case, for such a $j$, by the definition of $H_j$, $G^*[V_j] = G[V_j]$ must be $\frac12 \varepsilon$-far from any graph $Q$ on vertex set $V_j$ with $\phi(Q) \ge \phi^*_{in}$ (any such a graph $Q$ must be obtained from $G^*[V_i]$ by at least $\kappa_j > \frac12 \varepsilon d |V_j|$ addition/deletion of the edges), as required.
\end{proof}

We are now ready to prove Lemma \ref{lemma:partition-eps-far-improved}. We will set $\alpha_{\ref{lemma:partition-eps-far-improved}} = \min\{ \frac{c_{exp}}{150d} , \frac{1}{2 k c_{\ref{lemma:subset}}}\}$, and thus we have $\phi_{in}^* \le \frac{\varepsilon}{2 k c_{\ref{lemma:subset}}}$.

Our proof is by induction: we will construct a sequence of partitions $\{V_1\}$, $\{V_1, V_2\}, \dots, \{V_1, \dots, V_{k+1}\}$ of $V$ such that each partition $\{V_1, \dots, V_h\}$ satisfies the following properties:
\begin{enumerate}[(a)]
\item\label{case-a} $|V_i| \ge \frac{\varepsilon^2}{1152k} |V|$ for every $i$, $1 \le i \le h$, and
\item\label{case-b} $e(V_1, \dots, V_h) \le (h-1) \cdot c_{\ref{lemma:subset}} \cdot \phi_{in}^* \cdot d \cdot |V|$.
\end{enumerate}
%The final partition will satisfy the properties of Lemma \ref{lemma:partition-eps-far-improved}, for $c_{\ref{lemma:partition-eps-far-improved}}\ge  96ak (k+1)$ by conditions (\ref{case-a}) and (\ref{case-b}).

Our first partition is the trivial partition $\{V\}$, which clearly satisfies our properties. We then apply inductively Lemma \ref{lemma:partition}. Let us consider some partition $\{V_1, \dots, V_h\}$ with $1 \le h \le k$ and assume that this partition satisfies (\ref{case-a}) and (\ref{case-b}). We will show how to refine it to obtain a partition $\{V_1, \dots, V_{h+1}\}$ satisfying properties (\ref{case-a}) and (\ref{case-b}).

Let us first remove from $G$ all edges between pairs of all distinct sets $V_i$ and $V_j$, $1 \le i < j \le h$, to obtain a graph $G'$. Since $\phi_{in}^* \le \frac{\varepsilon}{2 k c_{\ref{lemma:subset}}}$, we have removed $e(V_1, \dots, V_h) \le (h-1) \cdot c_{\ref{lemma:subset}} \cdot \phi_{in}^* \cdot d \cdot |V| \le \frac12 \varepsilon d |V|$ edges from $G$, and therefore $G'$ is $\varepsilon/2$-far from $(k,\phi_{in}^*,\phi_{out}^*)$-clusterable and such that our partition satisfies the prerequisites of Lemma \ref{lemma:partition}.

Then, by Lemma \ref{lemma:partition}, there is a set $V_{i^*}$ with $1 \le i^* \le h$, such that $|V_{i^*}| \ge \frac{1}{8k} \cdot \frac{\varepsilon}{2} \cdot |V|$ and $G'[V_{i^*}] = G[V_{i^*}]$ is $\frac{\varepsilon}{4}$-far from any $H$ on vertex set $V_{i^*}$ with maximum degree $d$ and $\phi(H) \ge \phi_{in}^*$. Next, we apply Lemma \ref{lemma:subset} on $V_{i^*}$ to obtain a set $A \subseteq V_{i^*}$ with
%\Artur{I might be off, but my calculations are giving the bound I have here, that is, since ``$G[V_{i^*}]$ is $\frac{\varepsilon}{4}$-far from \dots '' we obtain $\frac{\varepsilon/4}{9} \cdot |V_i^*|$, rather than $\frac{\varepsilon/2}{9} \cdot |V_i^*|$ as in the earlier version.}
$\frac{\varepsilon/4}{18} \cdot |V_{i^*}| \le |A| \le \frac12 |V_{i^*}|$ such that $e(A, V_{i^*} \setminus A) \le c_{\ref{lemma:subset}} \phi_{in}^* d |V_{i^*}| \le c_{\ref{lemma:subset}} \phi_{in}^* d |V|$. This gives us our new partition $\{V_1, \dots, A, V_{i^*} \setminus A, \dots, V_h\}$.

Using the bound for the size of $V_{i^*}$, we have
%\Artur{This \textbf{doesn't work} here, some calculations need different constants, since as for now, $|A| \ge \frac{\varepsilon/4}{9} \cdot |V_i^*| \ge \frac{\varepsilon^2}{576k} \cdot |V|$, which is not exactly property (\ref{case-a})!!!}
$|A| \ge \frac{\varepsilon/4}{18} \cdot |V_{i^*}| \ge \frac{\varepsilon^2}{1152k} \cdot |V|$ and $|V_{i^*} \setminus A| \ge \frac12 |V_{i^*}| \ge \frac{\varepsilon}{32k} \cdot |V|$, and therefore by the induction hypothesis, our new partition satisfies (\ref{case-a}).

In order to prove (\ref{case-b}), we observe the following
\begin{eqnarray*}
    e(V_1, \dots, A, V_{i^*} \setminus A, \dots, V_h)
        & \le &
    e(V_1, \dots, V_{i^*}, \dots, V_h) +
        e(A, V_{i^*} \setminus A)
        \\
        & \le &
    (h-1) \cdot c_{\ref{lemma:subset}} \cdot \phi_{in}^* \cdot d \cdot |V| +
        c_{\ref{lemma:subset}} \cdot \phi_{in}^* \cdot d \cdot |V|
        \\
        & = &
    h \cdot c_{\ref{lemma:subset}} \cdot \phi_{in}^* \cdot d \cdot |V|
    \enspace,
\end{eqnarray*}
where the second inequality follows from our induction hypothesis and the bound above.

In summary, we have proven by induction the existence of a partition of $V$ into $k+1$ sets $V_1, \dots, V_{k+1}$ such that properties (\ref{case-a}) and (\ref{case-b}) are satisfied. Note that since property (\ref{case-b}) implies that for every $i$, $1 \le i \le k+1$, $e(V_i, V \setminus V_i) \le k \cdot c_{\ref{lemma:subset}} \cdot \phi_{in}^* \cdot d \cdot |V|$, we have
\begin{eqnarray*}
    \phi_G(V_i)
        & = &
    \frac{e(V_i, V \setminus V_i)}{d |V_i|}
        \, \le \,
    \frac{k \cdot c_{\ref{lemma:subset}} \cdot \phi_{in}^* \cdot d \cdot |V|}
            {d \cdot |V_i|}
        \, \le \,
    \frac{k \cdot c_{\ref{lemma:subset}} \cdot \phi_{in}^* \cdot |V|}
        {\frac{\varepsilon^2 |V|}{1152k}}
        \, = \,
    \frac{1152 \cdot k^2 \cdot c_{\ref{lemma:subset}}}{\varepsilon^2}
        \cdot \phi_{in}^*
    \enspace.
\end{eqnarray*}

Therefore, Lemma \ref{lemma:partition-eps-far-improved} follows by setting %$\alpha_{\ref{lemma:partition-eps-far-improved}} = \frac{1}{2k c_{\ref{lemma:subset}}}$ and
$c_{\ref{lemma:partition-eps-far-improved}} = 1152 \cdot k^2 \cdot c_{\ref{lemma:subset}}$.
\qed

%===========================================================================

\section{Conclusion}
\label{sec:conclusion}

We presented the first study of testing the clusterability of a graph in the bounded degree model, where we used both the inner conductance and outer conductance of a set to measure the quality of a cluster \cite{OT14:expander}. Our main result is an asymptotically optimal (up to polylogarithmic factors) algorithm with running time $\widetilde{O}(\sqrt{n}\cdot \poly({d,k,\varepsilon}))$ to test if a graph is $(k,\phi)$-clusterable or is $\varepsilon$-far from $(k, \phi^*)$-clusterable for $\phi^* = O_{d,k}(\frac{\phi^2 \varepsilon^4}{\log n})$. Our tester uses new ideas of testing pairwise closeness of distributions of random walks starting from a pair of sample vertices and draws from that conclusions on the graph structure. One of the key techniques underlying our analysis is a new application of the recent results on higher order Cheeger inequalities \cite{LOT12:high}.

For further research, one of the major open problem is to narrow the gap between $\phi$ and $\phi^*$, or to prove that the current gap is almost optimal for any tester with similar running time. As we discussed in Section \ref{subsection:Expansion}, fundamentally new ideas are needed here.

It would also be very interesting to gain deeper insights of the structure of graphs that are $\varepsilon$-far from $(k,\phi^*)$-clusterable, that is, to improve Lemma \ref{lemma:partition-eps-far-improved}. More specifically, is it possible to get rid of the dependency of $\varepsilon$ of the upper bounds for inner and/or outer conductance in Lemma \ref{lemma:partition-eps-far-improved}? %We believe such an improvement is possible since our current analysis inherently only takes into account the structure of graphs that are far from being a $\phi^*$-expander.

%============================================================

\bibliographystyle{alphabetic}
\bibliography{kclusterability}

%=============================================================================

\newpage
\appendix
\begin{center}\huge\bf Appendix \end{center}

%=============================================================================

\section{Useful tools from spectral graph theory}
\label{subsec:defn}

In this section, we introduce some useful tools from spectral graph theory that will be used in our analysis.

%=============================================================================

\subsection{Elementary facts from spectral graph theory}
\label{subsec:spectra}

Let $G = (V,E)$ be a weighted $d$-regular graph. Recall that we let $\A, \W = \frac{\I+\frac{1}{d}\A}{2}$, and $\LL = \I-\frac{1}{d} \A$ denote the adjacency matrix, the lazy random walk matrix and (normalized) Laplacian matrix of $G$, respectively.

Let $\1_S$ to denote the indicator vector of subset $S\subseteq V$, that is, $\1_S(v)=1$ if $v\in S$ and $\1_S(v)=0$ if $v\notin S$. We let $\1_v=\1_{\{v\}}$. For a vector $\p$, let $\p^T$ denote its transpose and let $\p(S):=\sum_{v\in S}\p(v)$. It is useful to notice that for any probability distribution $\p$ on $V$, $\p (\W)^t$ is the probability distribution of the endpoint of a length $t$ random walk with initial distribution $\p$. In particular, we let $\p_u^t:=\1_u(\W)^t$.

Let $0 = \lambda_1 \le \lambda_2 \le \dots \le \lambda_n \le 2$ be the eigenvalues of $\LL$ and let $\vv_1, \vv_2,\dots, \vv_n$ be the corresponding orthonormal left eigenvectors~\cite{Chu97:spectral}. Let $ \eta_1 \ge \eta_2 \ge \cdots \ge \eta_n$ denote the eigenvalues of $\W$, then it is easy to see that for each $i\le n$, $\eta_i=1-\frac{\lambda_i}{2}$ and $\vv_i$ is the corresponding eigenvector, where $\lambda_i$ and $\vv_i$ are the $i$th eigenvalue and eigenvector of $\LL$, respectively. Therefore, all the eigenvalues of $\W$ are non-negative and no larger than $1$. Note that since $\LL$ (or $\W$) is symmetric, its eigenvectors $\{\vv_i\}_{i=1,\dots, n}$ form an orthonormal basis of the Euclidean space $\mathbb{R}^V$. By the eigendecomposition of $\W$, we have $\W = \sum_{i=1}^n \eta_i \vv_i^T \vv_i = \sum_{i=1}^n (1 - \frac{\lambda_i}{2}) \vv_i^T \vv_i$.

We have the following basic fact.

\begin{fact}
\label{fact:spectra}
For any vertex $u$ and $t\ge 1$, we have
\begin{enumerate}
\item $\1_u = \sum_{i=1}^n \vv_i(u) \vv_i$,
\item $\sum_{i=1}^n \vv_i(u)^2 = 1$,
\item $\p_u^t = \1_u \W^t = \sum_{i=1}^n \vv_i(u) (1 - \frac{\lambda_i}{2})^t \vv_i$.
\end{enumerate}
\end{fact}

\begin{proof}
Since $\{\vv_i\}_{i=1,\dots, n}$ form an orthonormal basis of $\mathbb{R}^V$, we can represent $\1_u$ in terms of this basis, say $\1_u=\sum_{i=1}^n\alpha_i\vv_i$, where $\alpha_i\in \mathbb{R}$ for each $1\le i\le n$. By taking inner product with $\vv_i$ from both sides, we can solve $\alpha_i$ to get $\alpha_i=\langle\1_u, \vv_i\rangle = \vv_i(u)$, for any $i\le n$. Furthermore, $1 = \norm{\1_u}_2^2 = \sum_{i=1}^n \alpha_i^2 = \sum_{i=1}^n \vv_i(u)^2$, and $\1_u \W^t = (\sum_{i=1}^n \alpha_i \vv_i)(\sum_{i=1}^n (1-\frac{\lambda_i}{2}) \vv_i^T \vv_i)^t = \sum_{i=1}^n \vv_i(u) (1 - \frac{\lambda_i}{2})^t \vv_i$. This completes the proof of the fact.
\end{proof}

We also need the following simple fact of the eigenvalue $\lambda_i$ and eigenvector $\vv_i$ of the Laplacian $\LL$, which is known as the Rayleigh quotient formulation of $\lambda_i$ \cite{Chu97:spectral}.

\begin{fact}
\label{fact:eigenvalue}
For any $1 \le i \le n$, $\lambda_i = \frac{\vv_i (d\I-\A)\vv_i^T}{d\vv_i  \vv_i^T} = \frac{\sum_{(u,v) \in E}(\vv_i(u) - \vv_i(v))^2}{\sum_{u} d \vv_i^2(u)} = \frac{\sum_{(u,v) \in E}(\vv_i(u) - \vv_i(v))^2}{d}$.
\end{fact}

\begin{proof}%[Proof of Fact \ref{fact:eigenvalue}]
By definition, $\vv_i \LL = \lambda_i \vv_i$. Multiplying $\vv_i^T$ in both sides, we have $\vv_i \LL \vv_i^T = \lambda_i \vv_i \vv_i^T$, which gives that $\lambda_i = \frac{\vv_i (\I-\frac{1}{d}\A)\vv_i^T}{\vv_i  \vv_i^T} = \frac{\vv_i (d\I-\A)\vv_i^T}{d\vv_i\vv_i^T}$.

Now noting that for any vector $\vv$, $\vv (d\I) \vv^T = d \sum_u \vv_i(u)^2 = \sum_{(u,v)\in E} (\vv(u)^2 + \vv(v)^2)$, and $\vv \A \vv^T = \sum_{u,v: \A(u,v) \ne 0} \vv(u) \vv(v) = 2 \sum_{(u,v)\in E} \vv(u) \vv(v)$, we have $\vv(d\I-\A)\vv^T = \sum_{(u,v)\in E}(\vv(u)^2+\vv(v)^2-2\vv(u)\vv(v)) = \sum_{(u,v)\in E}(\vv(u)-\vv(v))^2$. Therefore, $\lambda_i = \frac{\vv_i (d\I-\A)\vv_i^T}{d\vv_i\vv_i^T} = \frac{\sum_{(u,v) \in E}(\vv_i(u)-\vv_i(v))^2}{\sum_{u}d\vv_i^2(u)} = \frac{\sum_{(u,v)\in E}(\vv_i(u)-\vv_i(v))^2}{d}$, where the last equation follows from the fact that $\vv_i$ is a unit-length vector for any $1\le i\le n$.
\end{proof}

%=============================================================================

\subsection{Volume-based definition of conductance and Cheeger's inequality}
\label{subsec:volume-definition}

In this section, we introduce the volume-based definition of conductance that has been used frequently in the literature before (cf. \cite{LOT12:high} and the references therein). In this section, we consider an arbitrary undirected and weighted graph $G=(V,E,w)$.

Let $w(v):=\sum_{(u,v)\in E}w(u,v)$ be the weighted degree of vertex $v$. For a vertex set $S\subseteq V$, let $w(S):=\sum_{v\in S}w(v)$ be the sum of weighted degrees of vertices in $S$. We will refer to $w(S)$ as the \emph{volume} of set $S$. For $S,T\subseteq V$, let $w(S,T):=\sum_{(u,v)\in E,u\in S, v\in T}w(u,v)$ be the sum of weights of edges with one endpoint in $S$ and the other endpoint in $T$. The volume-based conductance of $S$ in $G$ is defined as
\begin{displaymath}
    \phi_G^\vol(S)
        :=
    \frac{w(S,V\setminus S)}{w(S)}
    \enspace.
\end{displaymath}

Let $\phi^\vol(G):=\min_{S:w(S)\le w(V)/2}\phi_G^\vol(S)$. Note that generally, for a $d$-bounded degree graph $G$, the definition of conductance of a set $S$ we are using in the paper is slightly different from the volume-based definition of conductance of $S$ given as above. However, in a weighted $d$-regular graph $G=(V,E)$, these two definitions are identical.

We let $\A$ be the adjacency matrix of the weighted graph $G$, and let $\D$ denote the diagonal matrix with $\D(v,v)=w(v)$. Let $\LL = \I - \D^{-1/2} \A \D^{-1/2}$ denote the normalized Laplacian of $G$, and let $\lambda_i$ denote the $i$th smallest eigenvalue of $\LL$. Cheeger's inequality gives that

\begin{theorem}[\cite{AM85:lambda,Alo86:eigenvalues,SJ89:approximate}]
\label{thm:cheeger}
For any undirected and weighted graph $G$, it holds that
\begin{displaymath}
    \lambda_2/2
        \le
    \phi^\vol(G)
        \le
    \sqrt{2\lambda_2}
        \enspace.
\end{displaymath}
\end{theorem}

\junk{
Let $\rho_G^\vol(k)$ denote the minimum value of the maximum conductance over any possible $k$ disjoint nonempty subsets. That is,
\begin{displaymath}
    \rho_G^\vol(k)
        :=
    \min_{\textrm{disjoint $S_1,\dots, S_k$}}\max_{1\le i\le k}\phi_G^\vol(S_i)
        \enspace.
\end{displaymath}

Note that $\rho_G^\vol(2) = \phi^\vol(G)$. Lee et al.\ \cite{LOT12:high} prove the following higher-order Cheeger's inequality.
\begin{theorem}[\cite{LOT12:high}]
\label{thm:highcheeger-weighted}
For any weighted undirected graph $G$ and any $k \ge 2$, it holds that
\begin{displaymath}
    \lambda_k/2
        \le
    \rho_G^\vol(k)
        \le
    c_{\ref{thm:highcheeger}} k^2 \sqrt{\lambda_k}
        \enspace,
\end{displaymath}
where $c_{\ref{thm:highcheeger}}$ is some universal constant.
\end{theorem}
}

%=============================================================================

\section{On distribution testers: Proof of Lemma \ref{lem:collision}}
\label{subsec:proof-on-distribution}

For the sake of completeness, we give here a proof of Lemma \ref{lem:collision}.

\begin{proof}[Proof of Lemma~\ref{lem:collision}]
The description of the algorithm \textbf{$l_2^2$-norm tester} for testing if $\norm{\p_v^t}_2^2\le \sigma/4$ or $\norm{\p_v^t}_2^2>\sigma$ is very simple: \begin{enumerate}
\item let $Z_v$ denote the number of pairwise self-collisions of the $r$ samples from $\p_v^t$;
\item reject if and only if $Z_v\ge \frac12 \binom{r}{2} \sigma$.
\end{enumerate}

The performance of the above algorithm is guaranteed by the first paragraph of the proof of Lemma 4.2 in \cite{CS10:expansion} (that in turn is built on Lemma 1 in \cite{GR00:expansion}) by setting $\varepsilon = \frac12$ there. It is proven that if $r \ge 16\sqrt{n}$, with probability at least $1-\frac{16\sqrt{n}}{r}$, $\frac12 \binom{r}{2} \norm{\p_v^t}_2^2 \le Z_v \le \frac32 \binom{r}{2} \norm{\p_v^t}_2^2$. Therefore, with probability at least $1-\frac{16\sqrt{n}}{r}$, if $\norm{\p_v^t}_2^2\le \frac{\sigma}{4}$, then $Z_v\le \frac32 \binom{r}{2}\frac{\sigma}{4}<\frac12 \binom{r}{2} \sigma$ and the tester will accept; and if $\norm{\p_v^t}_2^2> \sigma$, then $Z_v\ge \frac12 \binom{r}{2}\sigma$ and the test will reject.
\end{proof}

%=============================================================================

\section{On tightness of Lemma \ref{lem:clusterable-eigenvector}}
\label{subsec:evidence-tight}

We prove the following lemma to show that Lemma \ref{lem:clusterable-eigenvector} is essentially tight for $k=2$ and constant $\phi_{in}$.

\begin{lemma}
\label{lem:counterexample}
Let $G = (V,E)$ be a $d$-regular graph composed of two parts $A$ and $B$, each of size $n/2$. Let $\phi_G(A) = \phi_G(B) = \phi_{out} \le \frac{1}{4d}$. Let $f := \vv_2$ be the second eigenvector with unit-length of the Laplacian matrix $\LL$ of $G$. Then
\begin{displaymath}
    \max\left\{
        \frac{1}{|A|} \sum_{u,v\in A}(f_u-f_v)^2,
        \frac{1}{|B|} \sum_{u,v\in B}(f_u-f_v)^2
    \right\}
        \ge
    \frac{\phi_{out}}{24d^3}
    \enspace.
\end{displaymath}
\end{lemma}

\begin{proof}
For any subset $U$, we define the potential of $U$ to be
$$\pot(U):=\frac{1}{|U|}\sum_{u,v\in U}(f_u-f_v)^2.$$
Let $x:=\frac{\phi_{out}}{\phi_{out}+12d^2(d+1)}\ge \frac{\phi_{out}}{24d^3}$. We will show that at least one of $\pot(A),\pot(B)$ is larger than $x$.

Assume on the contrary that $\pot(A),\pot(B)\le x$. We will derive a contradiction to the fact that $f$ is the second eigenvector of $\LL$.

First, for any subset $U$, we define the center of $U$ to be $\Delta_U := \frac{\sum_{u \in U} f_u}{|U|}$. Then we have
\begin{displaymath}
    \pot(U) = 2 \sum_{v \in U}(f_v - \Delta_U)^2
    \enspace.
\end{displaymath}

Now our assumption implies that
\begin{eqnarray}
    \sum_{u \in A}(f_u-\Delta_A)^2
        \le
    \frac{x}{2}
        \enspace,
        \quad
    \sum_{u \in B}(f_u-\Delta_B)^2
        \le
    \frac{x}{2}
    \enspace.
    \label{ineq:assumption}
\end{eqnarray}
Furthermore,
\begin{eqnarray}
    \sum_{u \in A}(f_u - \Delta_A)^2 + \sum_{u \in B}(f_u - \Delta_B)^2
        & = &
    \sum_{u \in V} f_u^2 - 2 \sum_{u \in A} \Delta_A f_u + |A| \Delta_A^2 - 2 \sum_{u \in B} \Delta_B f_u + |B| \Delta_B^2
        \nonumber
        \\
        & = &
    1 - \frac{n}{2} (\Delta_A^2 + \Delta_B^2)
        \nonumber
        \\
        & \le &
    x
        \enspace,
    \label{eqn:center-l2-sum}
\end{eqnarray}
where the penultimate equation follows from the fact that $\sum_{u\in A} f_u = |A|\Delta_A$, $|A|=|B|=n/2$ and $\sum_{u}f_u^2=1$ since $f$ is a unit vector.

On the other hand, since $f$ is the second eigenvector of $\LL$, then $\sum_{u}f_u=0$. Furthermore,
\begin{eqnarray}
    \Delta_A + \Delta_B
        =
    \frac{2}{n} (\sum_{u \in A} f_u + \sum_{u \in B} f_u)
        =
    0
    \enspace.
    \label{eqn:centersum}
\end{eqnarray}

Therefore, by inequality (\ref{eqn:center-l2-sum}) and equation (\ref{eqn:centersum}), we have that at least one of $\Delta_A,\Delta_B$ is positive. Wlog., we assume that $\Delta_A>0$. This further implies that $\Delta_A\ge\sqrt{\frac{1-x}{n}}$, and $\Delta_B\le-\sqrt{\frac{1-x}{n}}$.

Let $0<y\le 1$ that will be specified later. Let $A_1:=\{u\in A: (f_u-\Delta_A)^2> \frac{x}{2y|A|}\}$ and let $B_1:=\{u\in B: (f_u-\Delta_B)^2> \frac{x}{2y|B|}\}$. Then by our assumption of inequalities~(\ref{ineq:assumption}), we know that $|A_1|\le y|A|$ and $|B_1|\le y|B|$. We further define $A_2$ to be the subset in $A\setminus A_1$ such that for any $v\in A_2$, at least one of its neighbors is contained in $A_1$ or $B_1$. We define $B_2$ similarly. Since the maximum degree of vertices in $G$ is at most $d$, we know that $|A_2|+|B_2|\le d(|A_1|+|B_1|)$.

We call a vertex $v$ bad if $v$ belongs to $(A_1\cup A_2)\cup(B_1\cup B_2)$. Otherwise, we call $v$ good. Note that the number of bad vertices is equal to $|A_1\cup A_2|+|B_1\cup B_2|\le (d+1)(|A_1|+|B_1|)\le (d+1)y(|A|+|B|)=(d+1)yn$. Also, the number of edges involving any bad vertices is at most $d(d+1)yn$.

Now we let $y=\frac{\phi_{out}}{3(d+1)}$. Since the number of edges between $A$ and $B$ is $e(A,B)=\phi_{out}d|A|>d(d+1)yn$, there exists at least one edge, say $(u,v)\in E$, such that $u\in A$ and $v\in B$ and both $u,v$ are good.

Since $u$ is good, we know that all of its neighbors are in $A\setminus A_1$ or $B\setminus B_1$. Let $d_A, d_B$ denote the number of neighbors of $u$ belonging to $A\setminus A_1, B\setminus B_1$, respectively. By the fact that there exists at least one crossing edge $(u,v)$, we know that $d_B\ge 1$. Note that for any vertex $w \in A \setminus A_1$, $|f_w - \Delta_A| \le \sqrt{\frac{x}{2y|A|}} = \sqrt{\frac{x}{yn}}$, and for any vertex $w \in B \setminus B_1$, $| f_w - \Delta_B| \le \sqrt{\frac{x}{2y|A|}} = \sqrt{\frac{x}{yn}}$. We have that
\begin{eqnarray*}
    \sum_{w: (w,u) \in E} f_w
        & \le &
    (d-1)(\Delta_A + \sqrt{\frac{x}{yn}}) + \Delta_B + \sqrt{\frac{x}{yn}}
        \\
        & = &
    (d-2) \Delta_A + d \sqrt{\frac{x}{yn}}
%        \quad(\textrm{We specifiy $x = \frac{\phi_{out}}{\phi_{out} + 12 d^2 (d+1)}$ so that $2 d \sqrt{\frac{x}{yn}} = \sqrt{\frac{1-x}{n}} \le \Delta_A$})
        \\
        & \le &
    (d-\frac32) \Delta_A
        \\
        & < &
    d (1 - 2 \phi_{out}) (\Delta_A - \sqrt{\frac{x}{yn}})
%        \qquad(\textrm{We assume $\phi_{out}<\frac{1}{4d}$} here)
        \\
        & \le &
    d(1-\lambda_2)f_u
        \enspace,
\end{eqnarray*}
where the second inequality follows by our choices of $x$ and $y$ (since we set $x = \frac{\phi_{out}}{\phi_{out} + 12 d^2 (d+1)}$, we obtain $2 d \sqrt{\frac{x}{yn}} = \sqrt{\frac{1-x}{n}} \le \Delta_A$), the third inequality follows by our assumption that $\phi_{out} \le \frac{1}{4d}$, and in the last inequality we use the fact that $\lambda_2 \le 2 \phi(G) \le 2 \phi_{out}$.

Now since $f$ is the second eigenvector of $\LL$, that is, $f\LL=\lambda_2f$, we know that for each vertex $u$, $\sum_{w:(w,u)\in E}f_w=d(1-\lambda_2)f_u.$ This is a contradiction.
\end{proof}

\begin{remark}
Note that Lemma \ref{lem:counterexample} implies that if a graph is connected by two large clusters $A,B$, each of size $n/2$ and outer conductance $\phi_{out}$, then for at least one cluster, say $A$, the average value of $(\vv_2(u)-\vv_2(v))^2$ over all vertex pairs in $A$ is large. More precisely,
\begin{displaymath}
    \frac{1}{|A|^2}\sum_{u,v\in A}(\vv_2(u)-\vv_2(v))^2
        =
    \Omega\left(\frac{\phi_{out}}{d^3|A|}\right)
        =
    \Omega\left(\frac{\phi_{out}}{nd^3}\right)
    \enspace.
\end{displaymath}

Furthermore, if the inner conductance of each cluster is at least $\phi_{in}$ such that $\phi_{out}=\Theta(\frac{\phi_{in}^2}{\log n})$, then for any $t\le \Theta(\frac{\log n}{\phi_{in}^2})<\frac{1}{10\phi_{out}}$, we have
\begin{eqnarray*}
    \frac{1}{|A|^2}\sum_{u,v\in A}\norm{\p_v^t - \p_u^t}_2^2
        & = &
    \frac{1}{|A|^2} \sum_{u, v \in A} \sum_{i=1}^n (\vv_i(u)-\vv_i(v))^2 \left(1-\frac{\lambda_i}{2}\right)^{2t}
        \\
        & \ge &
    \frac{1}{|A|^2}\sum_{u,v \in A}(\vv_2(u)-\vv_2(v))^2 \left(1-\frac{\lambda_2}{2}\right)^{2t}
        \\
        & \ge &
    \Omega\left(\frac{\phi_{out}}{nd^3}\right) \cdot \left(1-\phi_{out}\right)^{2t}
        \\
        & = &
    \Omega\left(\frac{\phi_{out}}{nd^3}\right)
    \enspace,
\end{eqnarray*}
where the penultimate inequality follows from the inequality that $\lambda_2\le 2\phi_{out}$ and the last inequality follows from our choice of $t$.

Therefore, the average value of $\norm{\p_v^t - \p_u^t}_2^2$ over all vertex pairs $u,v$ in the cluster $A$ is $\Omega(\frac{\phi_{out}}{nd^3})$.
\end{remark}

%=============================================================================

\end{document}